\documentclass[journal,twoside,web]{ieeeconf}
\usepackage{cite}
\usepackage{amsmath,amssymb,amsfonts}
\usepackage{algorithmic}
\usepackage{graphicx}
\usepackage{textcomp}

\usepackage{mathtools} %←追加したパッケージ
\DeclareMathOperator*{\argmin}{arg\,min}%←追加したパッケージ
\usepackage{subcaption}
\usepackage{caption}
\usepackage{setspace}
\usepackage{url}
\usepackage{here}
\usepackage{booktabs}
\usepackage{multirow}
\usepackage{multicol}
\usepackage{bigstrut}
\usepackage{cases}
\usepackage{comment}
\usepackage {empheq}

\usepackage{pdfpages}
\newtheorem{theorem}{Theorem}
\newtheorem{lemma}{Lemma}

\newtheorem{coro}{Corollary}
\newtheorem{exam}{Example}
\newtheorem{rem}{Remark}

\newtheorem{assu}{Assumption}

\newcommand{\pd}{\partial}
\newcommand{\Nmn}{\mathcal{N}_{\mathrm{mn}}} 
\newcommand{\Nmj}{\mathcal{N}_{\mathrm{mj}}} 
\newcommand{\Cmn}{\mathcal{C}_{\mathrm{mn}}}
\newcommand{\Cmj}{\mathcal{C}_{\mathrm{mj}}}
\newcommand{\CC}{\mathcal{C}}
\newcommand{\CN}{\mathcal{N}}
\newcommand{\CE}{\mathcal{E}}
\newcommand{\CV}{\mathcal{V}}
\newcommand{\CA}{\mathcal{A}}
\newcommand{\CT}{\mathcal{T}}
\newcommand{\RA}{\mathrm{A}}
\newcommand{\RB}{\mathrm{B}}
\newcommand{\CNRA}{\mathcal{N}_\mathrm{A}}
\newcommand{\CNRB}{\mathcal{N}_\mathrm{B}}
\newcommand{\RK}{\mathrm{K}}
\newcommand{\CS}{\mathcal{S}}
\newcommand{\CU}{\mathcal{U}}
\newcommand{\BRdn}{\mathbb{R}^{d\times n}}

\begin{document}
% \includepdf[pages=1]{copyright.pdf}
\begin{figure*}
    \centering
    \includepdf{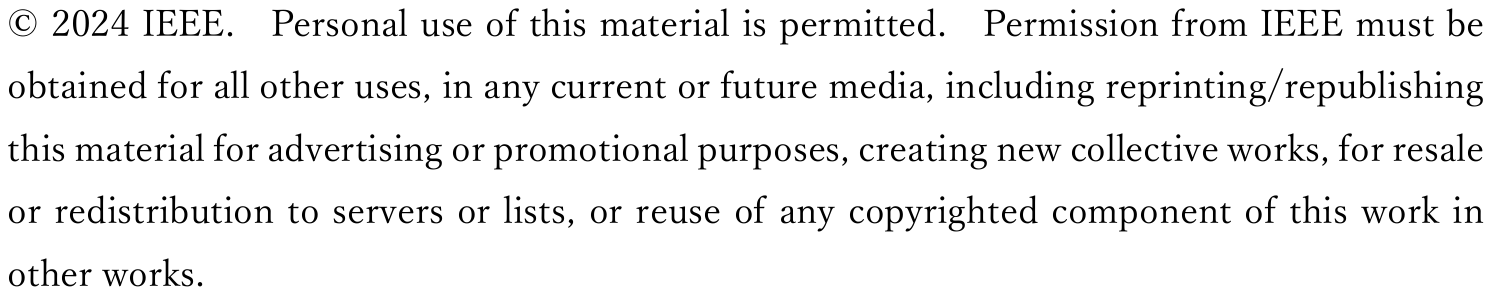}
\end{figure*}

\title{
Gradient-Based Distributed Controller Design Over Directed Networks
% Distributed Dynamic Matching of Two Groups of Agents with Different Sensing Ranges
}
\author{Yuto Watanabe, \IEEEmembership{Student Member, IEEE}, Kazunori Sakurama, \IEEEmembership{Member, IEEE}, and Hyo-Sung Ahn, \IEEEmembership{Senior Member, IEEE}
% \thanks{This paragraph of the first footnote will contain the date on 
% which you submitted your paper for review. It will also contain support 
% information, including sponsor and financial support acknowledgment. For 
% example, ``This work was supported in part by the U.S. Department of 
% Commerce under Grant BS123456.'' }
% \thanks{
% The next few paragraphs should contain 
% the authors' current affiliations, including current address and e-mail. For 
% example, F. A. Author is with the National Institute of Standards and 
% Technology, Boulder, CO 80305 USA (e-mail: author@boulder.nist.gov). }
% \thanks{S. B. Author, Jr., was with Rice University, Houston, TX 77005 USA. He is 
% now with the Department of Physics, Colorado State University, Fort Collins, 
% CO 80523 USA (e-mail: author@lamar.colostate.edu).}
\thanks{Yuto Watanabe and Kazunori Sakurama are with the Department of Systems Science, Graduate School of Informatics, Kyoto University, Yoshida-Honmachi, Sakyo-ku, Kyoto 606--8501, Japan, {\tt\small y-watanabe@sys.i.kyoto-u.ac.jp, sakurama@i.kyoto-u.ac.jp}.}
\thanks{Hyo-Sung Ahn is with the School of Mechanical Engineering, Gwangju Institute of Science and Technology (GIST), 123 Cheomdan-gwagiro, Buk-gu, Gwangju, 500-712 KOREA, {\tt\small hyosung@gist.ac.kr}.}
\thanks{
% A part of this work was supported by JSPS KAKENHI Grant Numbers 21H01352 and 22H01511.
% This work was partially supported by the joint project of Kyoto University and Toyota Motor Corporation, titled ``Advanced Mathematical Science for Mobility Society''.
This work was supported in part by JSPS KAKENHI Grant Numbers 21H01352, 22H01511, and the project of Theory of Innovative Mechanical Systems, collaborated by Kyoto University and Mitsubishi Electric Corporation.
}
% \thanks{
% © 2024 IEEE.  Personal use of this material is permitted.  Permission from IEEE must be obtained for all other uses, in any current or future media, including reprinting/republishing this material for advertising or promotional purposes, creating new collective works, for resale or redistribution to servers or lists, or reuse of any copyrighted component of this work in other works.
% }
}
\maketitle
% \renewcommand{\baselinestretch}{0.99}
% \thanks{
% © 2024 IEEE.  Personal use of this material is permitted.  Permission from IEEE must be obtained for all other uses, in any current or future media, including reprinting/republishing this material for advertising or promotional purposes, creating new collective works, for resale or redistribution to servers or lists, or reuse of any copyrighted component of this work in other works.
% }
\begin{abstract}
% In this paper, we propose a distributed controller design methodology for multi-agent systems over a class of directed interaction networks by extending the gradient-flow method, which is a popular design tool for distributed controllers but cannot be applied to directed networks.
% First, we show how to construct a distributed controller for systems over a class of time-invariant directed graphs.
% Next, we prove nice convergence properties and performance enhancement compared with the gradient-flow method.
% Moreover, as an application to time-varying networks, we apply the proposed controller to the dynamic matching problem of two agent groups with different sensing ranges, and sufficient conditions for successful matching are shown.
% This problem is a new coordination problem of making pairs from two agent groups to make the states of paired agents converge to the same value.
% Finally, numerical examples for systems over both time-invariant and time-varying networks illustrate the effectiveness of the proposed method.
In this study, we propose a design methodology of distributed controllers for multi-agent systems on a class of directed interaction networks by extending the gradient-flow method.
Although the gradient-flow method is a common design tool for distributed controllers, it is inapplicable to directed networks.
First, we demonstrate how to construct a distributed controller for systems over a class of time-invariant directed graphs.
Subsequently, we establish better convergence properties and performance enhancement than the conventional gradient-flow method.
To illustrate its application in time-varying networks, we address the dynamic matching problem of two distinct groups of agents with different sensing ranges.
This problem is a novel coordination task that involves pairing agents from two distinct groups to achieve a convergence of the paired agents' states to the same value.
Accordingly, we apply the proposed method to this problem and provide sufficient conditions for successful matching.
Lastly, numerical examples for systems on both time-invariant and time-varying networks demonstrate the effectiveness of the proposed method.
\end{abstract}

% \begin{IEEEkeywords}
% Distributed control, multi-agent systems, formation control, gradient-flow method, directed graphs
% % Enter key words or phrases in alphabetical 
% % order, separated by commas. For a list of suggested keywords, send a blank 
% % e-mail to keywords@ieee.org or visit \underline
% % {http://www.ieee.org/organizations/pubs/ani\_prod/keywrd98.txt}
% \end{IEEEkeywords}

\section{Introduction}\label{sec:introduction}

Owing to the advancement of information technology and the increase of large-scale systems, distributed control in multi-agent systems has garnered significant attention. This control strategy attempts to steer a multi-agent system towards a desired state through local information exchange, thereby increasing the system's scalability and fault tolerance. Because of these advantages, distributed control has been an important subject of extensive research for decades.

One of the most prevalent tools for designing distributed controllers is the gradient-flow method.
\textcolor{black}{For decades, this method has been actively utilized in
many studies on distributed control of single integrators, e.g., \cite{martinez2007motion,dimarogonas2008stability,mesbahi2010graph,oh2015survey,sun2016exponential,sakurama2014distributed,sakurama2020unified,sakurama2021generalized}.
This approach involves designing an objective function such that its gradient can be calculated in a distributed manner with the desired task achieved at its minimum points (or critical points).
% owing to its simple structure and expressive capabilities, 
% \textcolor{black}{For decades, the gradient-flow method has been actively utilized in
% many studies on distributed control of single integrators, e.g., \cite{martinez2007motion,dimarogonas2008stability,mesbahi2010graph,oh2015survey,sun2016exponential,sakurama2014distributed,sakurama2020unified,sakurama2021generalized}.} 
% The gradient-flow method provides a systematic design paradigm \cite{sakurama2020unified,sakurama2021generalized} for distributed controllers for various tasks by driving the states of agents into a designed set in which their task is achieved.
This method has good convergence properties not only for nonlinear systems \cite{khalil2014nonlinear} but also for hybrid systems \cite{bacciotti1999stability,shevitz1994lyapunov,haddad2008nonlinear}.
Remarkably, the gradient-flow method provides a systematic design paradigm \cite{sakurama2014distributed,sakurama2020unified,sakurama2021generalized} for distributed controllers for various tasks by driving the states of agents into a designed set in which their task is achieved.
Additionally, controller design strategies based on relative measurements via sensing have been developed for this method
in \cite{oh2015survey,sakurama2020unified,sakurama2021generalized}.
For example, the gradient-flow method can be implemented for formation control problems \cite{sakurama2021generalized,sakurama2020unified}
with only relative positions obtained by sensors (e.g., LiDAR).
}
Nevertheless, the gradient-flow method is not directly applicable to multi-agent systems over directed networks in general although directed graphs can express heterogeneity of the agents' specifications and asymmetric information flow, which often appear in practical applications.
This is because unidirectional and bidirectional edges are not distinguished in scalar-valued objective functions, which inhibits the distributedness of the gradient-based controllers.
\textcolor{black}{
Although distributed gradient-based control methods over directed graphs have been presented in \textcolor{black}{\cite{ren2005consensus,cao2007controlling,dorfler2010geometric,olfati2007consensus,oh2011distance,park2015stabilisation,pham2017distance,zhang2019control}}, their target tasks and graphs are limited to fairly simple ones, e.g., consensus, distance-based formation control, cycle graphs.
}
% Moreover, although
% many studies in the field of distributed optimization, e.g., \cite{nedic2014distributed,falsone2017dual,yang2019survey,xin2020general}, have vigorously addressed directed graphs, these methods are inferior to the gradient-flow method in terms of the simplicity and communication cost because distributed optimization methods force agents to introduce auxiliary variables.
% Moreover, although
% many studies in the field of distributed optimization, e.g., \cite{nedic2014distributed,falsone2017dual,yang2019survey,xin2020general}, have vigorously addressed directed graphs and have proposed versatile methods, these methods are inferior to the gradient-flow method in terms of the simplicity and communication cost because distributed optimization methods force agents to introduce auxiliary variables, which can require additional communication devices or sensors.
To the best of the authors' knowledge, no prior study has successfully extended the gradient-flow method to directed graphs, which complicates the systematic design of distributed controllers over directed graphs.

This study proposes a distributed controller design methodology for multi-agent systems over a class of directed graphs by extending the gradient-flow method.
First, we define a target class of time-invariant directed graphs, \textcolor{black}{which can express agents' different information-gathering abilities, two-layered leader-follower structures, and so forth.
Next, we present a design methodology of distributed controllers based on the gradient-flow method with a simple modification that enhances the performance.}
% The proposed controller is based on the gradient-flow method \cite{martinez2007motion,sakurama2014distributed} and feasible direction method \cite{zoutendijk1960methods}.
Subsequently, we analyze its convergence properties and conduct numerical experiments of distance-based formation control.
% The results indicate that the developed methodology can be applied to various tasks (e.g., formation control \cite{oh2015survey,sakurama2020unified}) and achieves a performance enhancement over the gradient-flow method.
Furthermore, as an application to time-varying networks, we consider the dynamic matching problem of two groups of agents with different sensing ranges. This problem is a novel coordination task, in which agents autonomously determine their matching partners while moving with local information to meet their partners positionally.
\textcolor{black}{
% This problem addresses dynamical systems, differently from the conventional approach for matching problem, i.e., combinatorial optimization \cite{gusfield1989stable,lovasz2009matching, papadimitriou1998combinatorial}.
% % Second, this study addresses this problem in a distributed fashion, which facilitates the  scalability and applicability.
% % Third, this problem can be seen as an extension of the assignment problem in this paper \cite{sakurama2020multi}.
% Although dynamical systems are considered in \cite{hata2011dynamical,ramli2015lyapunov},
%  the central management of the system is required in \cite{hata2011dynamical}, and symmetric communication between agents with homogeneous capabilities is assumed in \cite{ramli2015lyapunov}.
% Moreover, the dynamic matching problem can be seen as an extension of assignment problems \cite{sakurama2020multi} that aim to match agents with static objects or a prescribed shape of formation.
Possible applications of the dynamic matching problem include cooperation among multiple UAVs and UGVs \cite{saska2012cooperative,arbanas2018decentralized} and docking control of two teams of spacecraft \cite{woffinden2007navigating}.
}
In this task, we consider the directed proximity graph derived from two different sensing ranges as the sensing network.
Finally, through convergence analysis and numerical experiments, we demonstrate the effectiveness of the proposed method.

The key contributions of this paper are as follows:
(i) We propose an extension of the gradient-flow method to directed graphs while preserving important properties of the original gradient-flow method.
By utilizing unidirectional edges neglected in the gradient-flow method, we provide a theoretical guarantee of performance enhancement.
The class of directed graphs includes two layered leader-follower graphs and directed proximity graphs with two different sensing ranges.
Moreover,
the proposed method can be implemented with the same sensor or communication devices as the original gradient-flow method.
\textcolor{black}{
Note that one can deal with more general directed graphs, e.g., hierarchical graphs, by interconnecting systems with the proposed controller.
Moreover, by combining the proposed method with existing results on the gradient-flow method over directed graphs \cite{oh2011distance,park2015stabilisation,pham2017distance,zhang2019control}, the class assumed in this paper can be relaxed.
}
(ii) The developed methodology can be applied to a broad range of tasks over directed graphs.
% , including two layered leader follower graphs and systems with two different sensing ranges.
In the existing gradient-flow controllers over directed graphs in \cite{park2015stabilisation,olfati2007consensus,pham2017distance,zhang2019control}, their applicable tasks are limited to simple ones, such as consensus and distance-based formation control.
(iii) 
We demonstrate the applicability of the proposed method to 
a directed proximity graph with two different visions and subsequently show that it outperforms an existing method in numerical experiments. 
While most existing studies have been dedicated to undirected proximity graphs \cite{cortes2005spatially,laventall2009coverage,fan2011novel,ramli2015lyapunov,sakurama2020multi},
a few works treat directed proximity graphs for specific tasks, e.g.,
connectivity maintenance \cite{poonawala2017preserving,sabattini2013distributed} and navigation with collision avoidance \cite{sano2022decentralized}.
This study enables directed proximity graphs with two sensing ranges for various tasks.

This research is based on the conference paper of the authors \cite{watanabe2022matching}.
The additional contents are as follows:
(i) Further discussions, including additional theoretical results and the proofs of all theorems and lemmas.
(ii) New simulation results of distance-based formation control.
% (iii) Simulation results of dynamic matching with different sensing ranges newly conducted under several different initial states.
% Then, the effectiveness of the proposed method is verified in all cases.

The remainder of this paper is organized as follows.
In Section \ref{sec:pre}, preliminaries are presented.
% on the graph theory.
Section \ref{sec:setting} provides a problem setting and formulates the main problem.
In Section \ref{sec:digraph}, a distributed controller design methodology is developed for multi-agent systems over a class of directed graphs, and the favorable convergence properties are presented as a solution to the target problem.
Additionally, numerical examples of distance-based formation control demonstrate a nice performance of the proposed method.
% Next, in Section \ref{sec:main_result}, as a main result, a distributed controller for dynamic matching and a solution to the target problem are presented.
Next, in Section \ref{sec:application}, as an application to time-varying graphs, we apply the proposed method to the dynamic matching problem of two agent groups with different sensing ranges and then highlight the effectiveness of the proposed method via convergence analysis and numerical experiments.
% Section \ref{sec:proof} provides proofs of theorems. 
Finally, Section \ref{sec:conc} concludes the paper.

\section{Preliminaries}\label{sec:pre}
\subsection{Notations and Definitions}
Let $\mathbb{R}$ and $\mathbb{R}_+$ be the sets of real numbers and positive real numbers, respectively.
% For a set $\CS$, $\mathrm{cl}\CS$ denotes the closure of $\CS$.
Let $|\cdot|$ be the number of elements in a countable finite set and $\|\cdot\|$ represent the Euclidean and Frobenius norms of a vector and matrix, respectively.
For finite countable sets $\CA,\,\mathcal{B}$ such that $|\CA|\leq |\mathcal{B}|$, $\Pi(\CA,\mathcal{B})$ denotes the set of all one-to-one functions from $\CA$ to $\mathcal{B}$.
For functions $f_1,\ldots,f_n:\BRdn \to \mathbb{R}^d$ of a matrix variable $X\in \BRdn$ and a set $\mathcal{I}\subset\{1,\ldots,n\}$, %positive integers, 
let $[f_j(X)]_{j\in\mathcal{I}}\in \mathbb{R}^{d\times |\mathcal{I}|}$ be the collection of $f_j(X)$ corresponding to the indices of $\mathcal{I}$, i.e., $[f_j(X)]_{j\in \mathcal{I}}=[f_{i_1}(X),\ldots,f_{i_{|\mathcal{I}|}}(X)]$, where $i_1,\ldots,i_{|\mathcal{I}|} \in \mathcal{I}$ satisfy $1\leq i_1<\cdots<i_{|\mathcal{I}|}\leq n$.
For $X$ and a set $\CT\subset \BRdn$, their distance is defined as $\mathrm{dist}(X,\CT)=\inf_{Y\in\CT}\|X-Y\|$.
For a differentiable function $V:\BRdn\to\mathbb{R}$ of $X=[x_1,\ldots,x_i,\ldots,x_n]\in \BRdn$, 
let $\nabla_i V(X)=\pd V/\pd x_i(X)$.
For a nonnegative function $V:\BRdn\to\mathbb{R}$ and $\rho>0$, the set $L_V(\rho)=\{X\in\BRdn:V(X)\leq \rho\}$ is called \textit{the level set} of $V$ with respect to $\rho$.
If $V(X)\to\infty$ holds as $\|X\|\to\infty$, then $V(X)$ is said to be \textit{radially unbounded}.
\textcolor{black}{
For a function $f:\BRdn\to\mathbb{R}^l$, the zero set of $f(\cdot)$ is represented by $f^{-1}(0)=\{X\in\BRdn: f(X) = 0\}$. }
% % and for sets $\mathcal{D}_1,\,\mathcal{D}_2\subset \BRdn$, $\mathrm{dist}_{\mathrm{t}}(\mathcal{D}_1,\mathcal{D}_2)$ denotes their Hausdorff distance $\sup_{X\in \mathcal{D}_1}\mathrm{dist}(X,\mathcal{D}_2)$.
% \textcolor{black}{For $\CT\subset\BRdn$, let $\pd\CT$ denote the boundary of $\CT$}.
% For a set $\CC\subset \{1,\ldots,n\}$, $P_\CC(\mathcal{T})$ represents the projection of $\CT$ onto the $[x_j]_{j\in \CC}$-coordination space and is defined as $P_\CC(\mathcal{T})=\{Y\in \mathbb{R}^{d\times |\CC|}:\, \exists X \in \CT\, \mathrm{s.t.}\, Y=[x_j]_{j\in\CC} \}$ for $X=[x_1,\ldots,x_n]\in\BRdn$.

\subsection{Graph Theory}
\begin{figure}[t]
\centering
    \includegraphics[width=0.9\columnwidth]{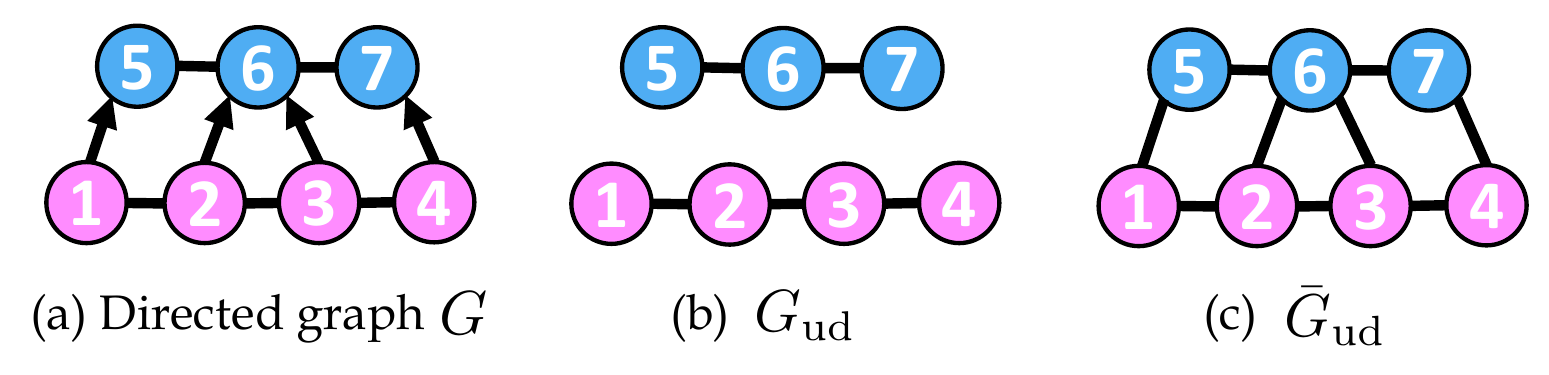}
   \caption{
  Examples of directed graph $G$ satisfying Assumption \ref{assu:g} in Example \ref{ex:assu_G}, undirected graph $G_\mathrm{ud}$ with the set of bidirectional edges in $G$, and undirected graph $\bar{G}_\mathrm{ud}$ with $\bar{\CE}_\mathrm{ud}$ in \eqref{def:ebar}.
  \textcolor{black}{
  This graph $G$ has a leader-follower structure.}
  }
  % \,G_\mathrm{B}$ for $G(X)$ in Fig.\ \ref{fig:ex_G} and $\CNRA,\,\CNRB$ in Example \ref{ex:N_K}.}
\label{fig:gagb}
\end{figure}

In this subsection, we provide graph-theoretic concepts.
Consider a graph $G=(\CN,\CE)$ with a node set $\CN=\{1,\ldots,n\}$ and an edge set $\CE$ comprising pairs $(i,j)$ of nodes $i,j\in \CN$.
For an edge $(i,j)$, we refer to $i$ and $j$ as the \textit{tail} and the \textit{head}, respectively.
In graph $G$, if both $(i,j)$ and $(j,i)$ are contained in $\CE$, each $(i,j),\,(j,i)$ is an \textit{bidirectional edge} of $G$.
On the other hand, if $(i,j)\in \CE$ but $(j,i)\notin \CE$, then $(i,j)$ is a \textit{unidirectional edge} of $G$.
If all edges of $G$ are bidirectional, $G$ is said to be an \textit{undirected graph}. Otherwise, we call $G$ a \textit{directed graph}.
\textcolor{black}{
Let $\CE_\mathrm{ud} \,(\subset\CE)$ be the set of all the bidirectional edges
and $\CE_\mathrm{di}\,(\subset\CE)$ be that of all the unidirectional edges.
% Hence, $\CE = \CE_\mathrm{ud}\cup\CE_\mathrm{di}$ and $\CE_\mathrm{ud}\cap\CE_\mathrm{di}=\emptyset$ hold.
% For $G=(\CN,\CE)$ and $\CE_\mathrm{ud}$, 
We define the undirected graph $G_\mathrm{ud}$ as $G_\mathrm{ud}=(\CN,\CE_\mathrm{ud})$, which is obtained by removing all the unidirectional edges of $G$.
In contrast, we define the undirected graph $\bar{G}_\mathrm{ud}$ as $\bar{G}_\mathrm{ud}=(\CN,\bar{\CE}_\mathrm{ud})$ with 
\begin{equation}
\label{def:ebar}
    \bar{\CE}_\mathrm{ud}=\CE\cup\{(i,j):(j,i)\in \mathcal{E}_{\mathrm{di}}\},
\end{equation}
which is obtained by replacing all unidirectional edges of $G$ with bidirectional ones.}
\textcolor{black}{
\begin{exam}\label{ex:GB}
Consider the directed graph $G$ in Fig.\ \ref{fig:gagb}a, where the arrows and lines correspond to unidirectional and bidirectional edges, respectively.
We obtain $\CE_\mathrm{di}=\{(1,5),(2,6),(3,6),(4,7)\}$ and $\CE_\mathrm{ud}=\CE\setminus\CE_\mathrm{di}$.
Hence, the undirected graph $G_\mathrm{ud}$ can be obtained by removing all the unidirectional edges $\CE_\mathrm{di}$ of $G$ as Fig.\ \ref{fig:gagb}b.
% \end{exam}
% \begin{exam}
Meanwhile, Fig.\ \ref{fig:gagb}c shows $\bar{G}_\mathrm{ud}=(\CN,\bar{\CE}_\mathrm{ud})$, where $\bar{\CE}_\mathrm{ud} = \CE\cup \{(5,1),(6,2),(6,3),(7,4)\}
$ holds from \eqref{def:ebar}.
Thus, by changing the unidirectional edges in $G$ into bidirectional edges, $\bar{G}_\mathrm{ud}$ can be obtained.
\end{exam}}

For $i\in\CN$, let $\CN_i(G)\subset \CN$ be the \textit{neighbor set} of node $i$ over $G=(\CN,\CE)$, defined as 
% \begin{equation*}
    $\CN_i(G)=\{j\in \mathcal{N}:(i,j)\in \mathcal{E}\}.$
% \end{equation*}

For $\delta>0$ and a matrix $X=[x_1,\ldots,x_n]\in\BRdn$, we define a \textit{$\delta$-proximity graph} as  $G_\delta(X)=(\CN,\CE_\delta(X))$ with
\begin{equation}\label{def:e_delta}
\CE_\delta(X)=\{(i,j):\|x_i-x_j\| \leq \delta,\,i,j\in \mathcal{N},\,i\neq j\}.
\end{equation}
Note that this graph is undirected. 
% \begin{exam}
%     The graph in Fig.\ \ref{fig:prox} is an example of $\delta$-proximity graphs. Each edge connects the agents within the distance $\delta$.
% \begin{figure}[t]
%  \centering
%  \includegraphics[width=0.42\columnwidth]{deltaprox_graph_clique.pdf}
%  \caption{Example of the $\delta$-proximity graph.}
%  \label{fig:prox}
% \end{figure}
% \end{exam}

For an undirected graph $G$,
we consider a set $\CC\subset\CN$.
For $\CC$ and $\CE$,  let $\CE|_{\CC}$ denote the subset of $\CE$ defined as $\mathcal{E}|_{\mathcal{C}}=\{(i,j)\in \mathcal{E}:i,j\in \mathcal{C}\}$. We call $G|_\CC=(\CC,\CE|_\CC)$ a subgraph induced by $\CC$.
If $G|_\CC$ is complete, $\CC$ is called a \textit{clique} in $G$.
If a clique $\CC$ is not contained by any other cliques, $\CC$ is said to be \textit{maximal}.
Let $\mathrm{clq}(G)$ be the set of all the maximal cliques in $G$. For $i\in\CN$, we define $\mathrm{clq}_i(G)$ as $\mathrm{clq}_i(G)=\{\mathcal{C}\in \mathrm{clq}(G): i \in \mathcal{C} \}$, which represents the set of maximal cliques containing $i$.
% that the node $i$ belongs to.
% \begin{exam}
%     Let the graph in Fig.\ \ref{fig:prox} be $G$.
%     Then, for $G$, $\mathrm{clq}(G)=\{\{1,2,3,4\},\{3,4,5\},\{5,6\}\}$ holds.
%     The maximal cliques containing node $3$ are $\mathrm{clq}_3(G)=\{\{1,2,3,4\},\{3,4,5\}\}$.
    
% \end{exam}
% \subsection{Stability theory}

%%%%%%%%%%%%%%%%%%%%%%%%%%%%%%%%%%%%%%%%%%%%%%%%%%%%%%%%%%%%%%%%%%%%%%%%%%%%%%%%%%%%%%%%%%%%%%%%%%%%%%%%%%%%%%%%%%%%%%%%%%%%%%%%%%%%%%%%%%%%%%%%%%%%%%%%%%%%%%%%%%%%%%%%%%%%%%%%%%%%%%%%%%%%%%%%%%%%

\section{Problem Setting}\label{sec:setting}
% \subsection{}

\textcolor{black}{
In this study, we address a controller design problem for a multi-agent system of single integrators over a certain class of directed graph $G$ based on the gradient-flow method, which is applicable to only undirected graphs in general.
Let $G_\mathrm{ud}$ be the undirected graph consisting of only the bidirectional edges of $G$, and let $V_\mathrm{ud}(X)$ be an objective function such that
the gradient-flow controller 
\begin{equation}
\label{gd_ex}
    \dot{x}_i= - \nabla_i V_\mathrm{ud}(X)
    % \tag{\text{\ref{gd_ex}}}
\end{equation}
is distributed with respect to $G_\mathrm{ud}$, where $X=[x_1,\ldots,x_n]$ and $x_1,\ldots,x_n$ are the states of the agents.
Then, the controller \eqref{gd_ex} is also distributed with respect to $G$.
% although unidirectional edges are ignored.
Now, we consider the control objective is to enforce $X$ to converge to the set of desirable states, represented by $\mathcal{T}$.
Then, the set $\nabla V_\mathrm{ud}^{-1}(0)$ of equilibrium points of \eqref{gd_ex} generally contains undesirable states, represented by $\nabla V_\mathrm{ud}^{-1}(0) \setminus \mathcal{T}$.
We can expect to reduce the undesirable equilibrium points by using unidirectional edges, ignored in \eqref{gd_ex}.
We will formulate this problem as designing a distributed controller such that its equilibrium set $\Omega$ satisfies  
\begin{equation}
\label{def:problem_eql}
     \mathcal{T} \subset \Omega \subset \nabla V_\mathrm{ud}^{-1}(0), %\tag{\text{\ref{def:problem_eql}}}
\end{equation}
and important properties of the gradient-flow methods are preserved, e.g., the global convergence to an equilibrium point and local attractiveness of the set $\mathcal{T}$.
Note that when \eqref{def:problem_eql} holds, the designed controller can perform better than the gradient-flow method because
the size of the undesired equilibrium points $\Omega\setminus\mathcal{T}$ of the former is smaller than that of the latter $\nabla V_\mathrm{ud}^{-1}(0) \setminus \mathcal{T}$, as shown in Fig. \ref{fig:problem1}.
}

\begin{figure}[t]
 \centering
 \includegraphics[width=0.7\columnwidth]{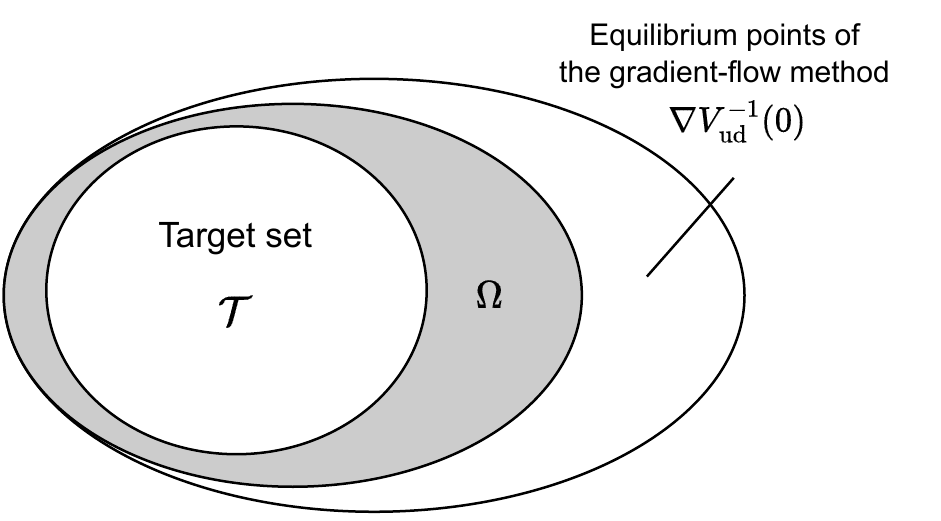}
 \caption{
 \textcolor{black}{
 Sketch of \eqref{def:problem_eql} in Problem 1.
 The undesired equilibrium set $\Omega\setminus\CT$ of the designed controller is expected to be smaller than that of the gradient-flow method $\nabla V_\mathrm{ud}^{-1}(0)\setminus \CT$.
 }
 }
 \label{fig:problem1}
\end{figure}

\subsection{Target System}\label{subsec:target_system}
We consider a multi-agent system consisting of $n$ agents in the $d$-dimensional space.
Let $\CN=\{1,\ldots,n\}$ be the set of agent indices.
The dynamics of agent $i\in \CN$ is given as 
\begin{align}\label{dynamics}
\dot{x}_i(t)=u_i(t),    
\end{align}
where $x_i(t)\in\mathbb{R}^d$ and $u_i(t)\in\mathbb{R}^d$ denote the state and control input, respectively.
The communication (or sensing) network of the agents is expressed as a time-invariant directed graph $G=(\CN,\CE)$ with edge set $\CE$.
Agent $i$ can obtain the states $x_j,\,j\in\CN_i$ of its neighbors.
% via communication or sensing.

Next, we consider a state feedback controller for agent $i$ with a function $f_i:\BRdn\rightarrow \mathbb{R}^d$ as
\begin{align}\label{controller}
    u_i(t)=f_i(X(t)),
\end{align}
where $X(t)=[x_1(t),\ldots,x_n(t)]\in \BRdn$. Then, $f_i(X)$ must depend only on the neighboring agents' states, i.e., it must be of the form $f_i(X)=\hat{f}(x_i,[x_j]_{j\in\CN_i(G)})$ with a function $\hat{f}:\mathbb{R}^{d\times (|\CN_i(G)|+1)}\rightarrow \mathbb{R}^d$.
This type of controller is called a \textit{distributed controller} with respect to $G$.

In this study, we assume that graph $G$ satisfies the following assumption.
\textcolor{black}{Note that the proposed method can be extended to more general graphs, e.g, hierarchical graphs, in a straightforward way, as shown in Remarks \ref{rem:extension} and \ref{rem:directed_extension}.}
% we consider a situation where sensing or communication capabilities differ among agents.
% Accordingly, we assume that graph $G$ satisfies the following assumption.
% \textcolor{black}{
% For simplicity,
% we now assume that directed graph $G$ satisfies the following assumption.
% Note that the core idea of the proposed controller is valid for more general directed graphs, as shown in Remark **.
% }
\begin{assu}\label{assu:g}
Graph $G$ satisfies $\CV_{\mathrm{t}}\cap\CV_{\mathrm{h}}=\emptyset$ for 
\begin{equation}
    \label{def:v1}
\CV_{\mathrm{t}}:=\{i\in\CN:\exists j\in \CN\,\mathrm{s.t.}\,(i,j)\in\CE_\mathrm{di}\}
\end{equation}
\begin{equation}
    \label{def:v2}
\CV_{\mathrm{h}}:=\{i\in\CN:\exists j\in \CN\,\mathrm{s.t.}\,(j,i)\in\CE_\mathrm{di}\},
\end{equation}
where $\CE_\mathrm{di}$ is the set of all unidirectional edges in $\CE$.
\end{assu}

Now, set $\CV_{\mathrm{t}}$ comprises the nodes that are tails of unidirectional edges, while $\CV_{\mathrm{h}}$ comprises the nodes that are heads.
Thus, this assumption implies that there is no node of $G$ that is the tail and head of distinct unidirectional edges in $\CE$.
Intuitively, this implies that agents in $\CV_{\mathrm{t}}$ can unilaterally observe those in $\CV_{\mathrm{h}}$, but not vice versa.
In other words, no agent simultaneously becomes a unilateral observer and a unilaterally observed agent.
\textcolor{black}{
As examples of graphs in Assumption \ref{assu:g}, we have leader-follower graphs in Example \ref{ex:assu_G} as well as directed proximity graphs in Section \ref{sec:application}.
}
\textcolor{black}{\begin{exam}\label{ex:assu_G}
    The directed graph $G$ in Fig.\ \ref{fig:gagb}a with a leader-follower structure satisfies Assumption \ref{assu:g}.
    % , and circles with dotted lines represent the area where each agent can observe others, i.e., the sensing range.
    % Now, the pink agents in Fig.\ \ref{fig:ex_G} have a larger sensing range than the blue ones, meaning that the system consists of two types of agents with different sensing capabilities, differently from $\delta$-proximity graphs with a single sensing range such as Fig.\ \ref{fig:prox}.
    Here, since $\CE_\mathrm{di}=\{(1,5),(2,6),(3,6),(4,7)\}$, we obtain $\CV_\mathrm{t}=\{1,2,3,4\}$ (in pink) and $\CV_{\mathrm{h}}=\{5,6,7\}$ (in blue), 
    which satisfy  $\CV_{\mathrm{t}}\cap\CV_{\mathrm{h}}=\emptyset$.
    % Hence, $G(X)$ in Fig.\ \ref{fig:ex_G} satisfies Assumption \ref{assu:g}.
\end{exam}}
\subsection{Control Objective}\label{subsec:control_objective}

For the multi-agent system in \eqref{dynamics}, we consider a target set $\CT \neq \emptyset$, which represents the desired configuration of agents, i.e., $x\in\CT$ represents the achievement of the desired task.
Hence, it is expected to achieve
\begin{equation}\label{def:objective2}
    \lim_{t\to \infty} \mathrm{dist}(X(t),\CT)= 0.
\end{equation}
% where $X(t)$ is the trajectory of the system in \eqref{dynamics} with the controller.
In the following, we describe \eqref{def:objective2} as $X(t)\to \CT$.
\textcolor{black}{If $f_i(X_0)=0$ holds for all $i\in\CN$ with \eqref{dynamics} and \eqref{controller}, $X_0\in\BRdn$ is said to be an \textit{equilibrium point}.}
The set $\CT$ is said to be \textit{locally attractive} if there exists an open set $\CA$ containing $\CT$ such that $X(0)\in\CA \Rightarrow X(t)\to\CT$. 
In addition, if $X(t)\to\CT$ holds for any $X(0)\in\BRdn$, $\CT$ is said to be \textit{globally attractive}.

The gradient-flow method is a conventional tool used to design distributed controllers for multi-agent systems.
A controller using this method is designed as
\begin{equation}\label{gd}
    f_i(X) = - \nabla_i V(X),
\end{equation}
where $V:\BRdn\to\mathbb{R}_+\cup\{0\}$ is a continuously differentiable nonnegative objective function.
By assigning $\nabla V^{-1}(0)$ as $\CT = \nabla V^{-1}(0)$,
we can guarantee $X(t) \to \nabla V^{-1}(0) = \CT$
by the controller in \eqref{gd}.
\textcolor{black}{
% However, guaranteeing $\CT =\nabla V^{-1}(0)$ is often challenging for designing distributed controllers.
\textcolor{black}{
However, it is usually not possible to choose an objective function such that $\CT = \nabla V^{-1}(0)$ unless $G_\mathrm{ud}$ is sufficiently dense.}
Instead, we can ensure $\CT \subset \nabla V^{-1}(0)$, which is a necessary condition of the task achievement.}
Accordingly, the set $\nabla V^{-1}(0) \setminus \CT$ is an \textit{undesired equilibrium set} of the system given by \eqref{dynamics} with \eqref{controller} and \eqref{gd}.
% Note that designing a gradient-distributed $V$ satisfying $\CT \subset V^{-1}(0)$ or $\CT \subset \nabla V^{-1}(0)$ is generally possible but requiring $\CT = V^{-1}(0)$ is often difficult depending on $G$.
We expect to design $V(X)$ such that $\nabla V^{-1}(0)\setminus \CT$ is maximally small. 
% as small as possible.

To design a distributed controller via the gradient-flow method, $V(X)$ needs to satisfy 
% \begin{equation}\label{def:gradient-distributed}
 $\nabla_i V(X)=\hat{f}(x_i,[x_j]_{j\in\CN_i(G)})$
% \end{equation} 
for all $i\in\CN$ with a function $\hat{f}:\mathbb{R}^{d\times (|\CN_i(G)|+1)}\rightarrow \mathbb{R}^d$.
Such a function $V(X)$ is said to be \textit{gradient-distributed} with respect to $G$.
The controller in \eqref{gd} with a gradient-distributed $V(X)$ becomes a distributed controller.
\textcolor{black}{
The gradient-flow method is, however, mainly for undirected graphs because of
% $V(X)$ is not gradient-distributed with respect to directed graphs in general, which attributes to the fact that undirected and unidirectional edges are not distinguished in taking the gradient of $V(X)$.
% Also, when we design the gradient prior to $V(X)$, 
the requirement concerning symmetry of the Hessian matrix of $V(X)$.
% the objective function is not satisfied for directed graphs, which imposes significant difficulty on the convergence guarantee.
}
One of the strategies to apply the gradient-flow method to directed graphs is to remove all the unidirectional edges.
 This means considering a new objective function $V_\mathrm{ud}(X)$ that is gradient-distributed with respect to the undirected graph $G_\mathrm{ud}=(\CN,\CE_\mathrm{ud})$ consisting of $G$'s bidirectional edges.
% Let $V_\mathrm{ud}(X)$ an objective function based on $G_\mathrm{ud}=(\CN,\CE_\mathrm{ud})$.
This strategy forces us to abandon all the information received via unidirectional edges.
This poses a question; \textit{is it possible to reflect the information via unidirectional edges in the gradient-flow method with $V_\mathrm{ud}(X)$ to enhance its performance?}
% the gradient-flow method can only be applied to undirected graphs.
% For directed graphs, the stability of the system
% cannot be guaranteed because its time derivative $\dot{V}(X(t))$ does not always satisfy $\dot{V}(X(t))\leq 0$ due to the asymmetry of the Hessian matrix of $V(X)$ even if the objective function $V$ is gradient-distributed with regard to $G$.
% One of the strategies to apply the gradient-flow method to multi-agent systems over directed graphs is to use an undirected subgragh of $G$, i.e., to consider a new objective function $V_\mathrm{ud}(X)$ such that its gradient is distributed with respect to undirected graph $G_\mathrm{ud}=(\CN,\CE_\mathrm{ud})$, where
% $\CE_\mathrm{ud}$ is the set of all bidirectional edges in $G$.
% In other words, we have to ignore all unidirectional edges.
% Then, using $G_\mathrm{ud}$ instead of $G$, i.e., ignoring the unidirectional edges in $G$ results in a waste of information.
% This poses a question; how we reflect unidirectional edges in the gradient-flow controller without sacrificing information?

To answer the aforementioned question and circumvent the disadvantage of the gradient-flow method to directed graphs, we address the following problem in this study.
% \textcolor{black}{Here, by \eqref{dynamics} and \eqref{controller}, $X\in\BRdn$ is said to be an \textit{equilibrium point} if $f_i(X)=0$ holds for all $i\in\CN$.}
\paragraph*{Problem 1}
% Consider a set of time-invariant directed graphs with the node set $\CN$.
Consider a time-invariant directed graph $G$ satisfying Assumption \ref{assu:g}.
% Let $G_\mathrm{ud}=(\CN,\CE_\mathrm{ud})$ be an induced undirected subgraph of $G$.
Consider a function $V_\mathrm{ud}:\BRdn\to\mathbb{R}$ that satisfies the following conditions:
\begin{itemize}
    % \item[B$1$)] $V_\mathrm{ud}$ is nonnegative and continuously differentiable.
    \item[C$1$)] $V_\mathrm{ud}(X)$ is gradient-distributed with respect to the undirected graph $G_\mathrm{ud}$.
    % with $\CE_\mathrm{ud}$ satisfying \eqref{condition:eb}.
    \item[C$2$)] $\CT\subset \nabla V_\mathrm{ud}^{-1}(0)$ is satisfied for the target set $\CT$.
%     \begin{equation}
%         \label{def:Omega_B}
% \nabla V_\mathrm{ud}^{-1}(0) := \{X\in \mathbb{R}^{d\times n}: \nabla_i V_\mathrm{ud}(X)=0,\:\forall i\in \CN\}.
%     \end{equation}
%     % \item[B$4$)] The level set $L_{V_\mathrm{ud}}(\rho)=\{X:\in\BRdn:V_\mathrm{ud}(X)\leq \rho\}\subset\BRdn$ of $V_\mathrm{ud}$ is non-empty and bounded for any $\rho>0$.
\end{itemize}
% continuously differentiable, nonnegative, and gradient-distributed objective function $V_\mathrm{ud}:\BRdn\to\mathbb{R}$ satisfying $\CT\subset V_\mathrm{ud}^{-1}(0)$.
% Assume that the level set $L_{V_\mathrm{ud}}(\rho)=\{X:\in\BRdn:V_\mathrm{ud}(X)\leq \rho\}\subset\BRdn$ of $V_\mathrm{ud}$ is non-empty and bounded for any $\rho>0$.
Then, for the system in \eqref{dynamics} with \eqref{controller}, design a distributed controller $f_i$ for all $i\in\CN$, such that the set $\Omega\subset\BRdn$ of equilibrium points of the system is globally attractive and satisfies \eqref{def:problem_eql}.
    % \begin{equation}
    % \label{def:problem_eql}
    %     \CT \subset \Omega  \subset \nabla V_\mathrm{ud}^{-1}(0).
    % \end{equation}
%     % $X(t)\to \Omega\cap L_{V_\mathrm{ud}}(\rho)$ is satisfied.
%     \item[b)] 
%     % for any initial state $X(0)\in L_{V_\mathrm{ud}}(\rho)$, there exists an 
%     the set $V_\mathrm{ud}(0)^{-1}$ is locally attractive.
% \end{itemize}

% Accordingly, we have two comments on Problem $1$.
We can construct $V_\mathrm{ud}(X)$ satisfying the conditions C1 and C2 with conventional works such as \cite{martinez2007motion,oh2015survey,sakurama2014distributed,sakurama2020unified,sakurama2021generalized}.
\textcolor{black}{If \eqref{def:problem_eql} holds, the performance of the gradient-flow method is enhanced concerning its undesired equilibrium points, \textcolor{black}{as shown in Fig.\ \ref{fig:problem1}.}
% by using unidirectional edges
% This is because when \eqref{def:problem_eql} is satisfied, the undesired equilibrium points $\Omega\setminus\CT$ generated by the designed controller is smaller than $\nabla V_\mathrm{ud}^{-1}(0) \setminus\CT$, i.e.,
% those of the gradient-flow controller in \eqref{gd} with $V_\mathrm{ud}$.
}
% Note that it is desirable that the set $\Omega$ is maximally small.
%%%%%%%%%%%%%%%%%%%%%%%%%%%%%%%%%%%%%%%%%%%%%%%%%%%%%%%%%%%%%%%%%%%%%%%%%%%%%%%%%%%%%%%%%%%%%%%%%%%%%%%%%%%%%

\section{Main Result
}\label{sec:digraph}

This section presents a distributed controller for multi-agent systems over the directed graphs satisfying Assumption \ref{assu:g}.
This methodology can be applied to various tasks, including formation control and dynamic matching.

\subsection{Controller Design}\label{subsec:controller design}

% As a preliminary, for $G$ satisfying Assumption \ref{assu:g}, we introduce an undirected graph $\bar{G}_\mathrm{ud}=(\CN,\bar{\CE}_\mathrm{ud})$ with 
% \begin{equation}
% \label{def:ebar}
%     \bar{\CE}_\mathrm{ud}=\CE\cup\{(i,j):(j,i)\in \mathcal{E}_{\mathrm{di}}\},
% \end{equation}
% which is the undirected graph obtained by replacing all unidirectional edges of $G$ with undirected ones.
To design a controller for the target set $\CT$, we consider the undirected graph $\bar{G}_\mathrm{ud}=(\CN,\bar{\CE}_\mathrm{ud})$ with $\bar{\CE}_\mathrm{ud}$ in \eqref{def:ebar}. Let $\bar{V}_\mathrm{ud}:\BRdn\to\mathbb{R}$ be an objective function satisfying the following conditions:
\begin{itemize}
    % \item[A$1$)] $V_\RA$ is nonnegative and continuously differentiable.
    \item[C$3$)] $\bar{V}_\mathrm{ud}$ is gradient-distributed with respect to $\bar{G}_\mathrm{ud}=(\CN,\bar{\CE}_\mathrm{ud})$.
    % in \eqref{def:ebar}.
    \item[C$4$)] $\CT\subset \nabla\bar{V}_\mathrm{ud}^{-1}(0)$ holds.
    % \item[A$4$)] The level set $L_{V_\RA}(\rho)=\{X:\in\BRdn:V_\RA(X)\leq \rho\}\subset\BRdn$ of $V_\RA$ is non-empty and bounded for any $\rho>0$.
\end{itemize}
\textcolor{black}{$\bar{V}_\mathrm{ud}$ can be designed similarly to designing the function $V_\mathrm{ud}$.}
Note that applying the gradient-flow method to $\bar{V}_\mathrm{ud}(X)$ directly is generally impossible because $\bar{V}_\mathrm{ud}(X)$ is not gradient-distributed with respect to $G$.
% because these conditions C$3$ and C$4$ can be seen as the conditions C$1$ and C$2$ with regard to $\bar{V}_\mathrm{ud}$ and $\bar{G}_\mathrm{ud}$ on behalf of $V_\mathrm{ud}$ and $G_\mathrm{ud}$, respectively.

\textcolor{black}{
Our goal in controller design is to improve the performance of the baseline gradient-flow method in \eqref{gd} with $V(X)=V_\mathrm{ud}(X)$ in the sense of \eqref{def:problem_eql} by utilizing unidirectional edges and the function $\bar{V}_\mathrm{ud}(X)$.
% Using both $V_\mathrm{ud}$ and $\bar{V}_\mathrm{ud}$, 
To this end, we propose the following gradient-based controller:}
\begin{align}
\label{ctr:digraph}
%\small
 f_i(X) = \begin{cases}
  -\bar{g}_i(X) -\kappa_{i}\displaystyle\nabla_i V_\mathrm{ud}(X) , & i \in \CV_\mathrm{t} \\
  -\mu_{i}\displaystyle\nabla_i V_\mathrm{ud}(X) , & i\in \CN\setminus\CV_\mathrm{t},
 \end{cases}
\end{align}
with constants $\kappa_i \geq 0, i\in\CV_\mathrm{t}$ and $\mu_{i}>0,\,i\in \CN\setminus\CV_\mathrm{t}$.
Here, the function $\bar{g}_i:\BRdn\to\mathbb{R}^d$ is given as a solution to the following minimization problem:
\begin{subequations}\label{gibar_all}
\begin{align}
\label{gibar}
\small
 &\min_{g_i\in \mathbb{R}^d} 
 \left\| g_i-\frac{1}{2}\left(\lambda_{i} \nabla_i  \bar{V}_\mathrm{ud} (X) + \eta_{i}\nabla_i V_\mathrm{ud} (X)\right)
 \right\|^2
 % + \left\|g_i-\mu_{\mathrm{B}i}\nabla_i V_\mathrm{ud} (X) \right\| 
 \\
 \label{gibar:c}
 &\mathrm{s.t.}\quad
 g_i^\top \nabla_i  \bar{V}_\mathrm{ud} (X) \geq 0,\:
 g_i^\top \nabla_i V_\mathrm{ud} (X) \geq 0,
\end{align}
\end{subequations}
where $\lambda_i,\,\eta_i>0$ for all $i\in\CV_\mathrm{t}$.
% Note that 
Then,  $\bar{g}_i(X)$ can be obtained explicitly as follows:
\begin{subequations}\label{gbar_explicit}
\small
\begin{empheq}[left={\bar{g}_i(X)= \empheqlbrace \,}]{alignat=6}
 % \begin{dcases}
 & \hat{g}_i(X), \label{gbar_a}\\
 &\quad 
 % \begin{split}
     \!\text{if}\begin{cases}
      \hat{g}_i(X)^\top \textstyle\nabla_i \bar{V}_\mathrm{ud}(X)\geq 0,\\
\hat{g}_i(X)^\top \textstyle\nabla_i V_\mathrm{ud}(X)\geq 0 
     \end{cases}\nonumber
 % \end{split}
 \\
 &\hat{g}_i(X)- \frac{\hat{g}_i(X)^\top\nabla_i \bar{V}_\mathrm{ud}(X)}{ \bigl\|\nabla_i \bar{V}_\mathrm{ud}(X)\bigr\|^2}\nabla_i \bar{V}_\mathrm{ud} (X), 
 \label{gbar_b}
 \\
 &\quad\!
     \text{if}\begin{cases}
      \hat{g}_i(X)^\top \textstyle\nabla_i \bar{V}_\mathrm{ud}(X)< 0,\\
\hat{g}_i(X)^\top \textstyle\nabla_i V_\mathrm{ud}(X)\geq 0 
     \end{cases}\nonumber
\\
  &\hat{g}_i(X) - \frac{\hat{g}_i(X)^\top\nabla_i V_\mathrm{ud}(X)}{ \bigl\|\nabla_i V_\mathrm{ud}(X)\bigr\|^2}
  \nabla_i V_\mathrm{ud} (X) \label{gbar_c}
  % \nabla_i V_\mathrm{ud} (X)
  ,\\
 &\quad
     \!\text{if}\begin{cases}
      \hat{g}_i(X)^\top \textstyle\nabla_i \bar{V}_\mathrm{ud}(X)\geq 0,\\
\hat{g}_i(X)^\top \textstyle\nabla_i V_\mathrm{ud}(X)< 0,
     \end{cases} \nonumber
 % \end{dcases}\normalsize
\end{empheq}
\end{subequations}
where $\hat{g}_i(X)=\left(\lambda_{i}\nabla_i \bar{V}_\mathrm{ud}(X) + \eta_i \nabla_i V_\mathrm{ud}(X) \right)/2$.
\textcolor{black}{
% Note that $\kappa_i$, $\lambda_i$, and $\mu_i$ are tunable parameters for adjusting 
% the size of the control input and contributions of $\nabla_i \bar{V}_\mathrm{ud}(X)$ and $\nabla_i V_\mathrm{ud}(X)$ to the input.
% In addition, although the function $\bar{g}_i(\cdot)$ in \eqref{gbar_explicit} is piece-wise defined, $\bar{g}_i(X)$ is bounded as long as $\nabla_i \bar{V}_\mathrm{ud}(X)$ and $\nabla_i V_\mathrm{ud}(X)$ take finite values.
The constants $\kappa_i$ and $\mu_i$ are the control gains with respect to the gradient $\nabla_i V_\mathrm{ud}(X)$ for $i\in \mathcal{V}_\mathrm{t}$ and $i\in \mathcal{N}\setminus\mathcal{V}_\mathrm{t}$.
The gain $\kappa_i$ in \eqref{ctr:digraph} needs to be positive to guarantee $\dot{V}_\mathrm{ud}(X)\to 0$, except for a special case in Corollary \ref{coro:eql}.
On the other hand, $\lambda_i$ and $\eta_i$ in $\bar{g}_i(X)$ in \eqref{gibar} can be used to adjust the contributions of $\nabla_i \bar{V}_\mathrm{ud}(X)$ and $\nabla_i V_\mathrm{ud}(X)$ to $\bar{g}_i(X)$ in \eqref{gibar}.
If we set $\lambda_i\gg \eta_i$, information from unidirectional edges strongly influences on the input, and vice versa.
}

% The norm in \eqref{gibar} indicates the difference between $g_i$ and the average between $\lambda_{i} \nabla_i  \bar{V}_\mathrm{ud} (X)$ and $\mu_{i}\nabla_i V_\mathrm{ud} (X)$.
% The constraints in \eqref{gibar:c} imply that $-g_i$ must point to a descent (non-ascent) direction of $V_\mathrm{ud}(X)$ and $\bar{V}_\mathrm{ud}(X)$ in terms of $x_i$.

For \eqref{gibar_all}, the function $\bar{g}_i(X)$ is given as the orthogonal projection of $\hat{g}_i(X)$ onto the region where \eqref{gibar:c} is satisfied.
This implies that $\bar{g}_i(X)$ enables us to decrease the additional objective function $\bar{V}_\mathrm{ud}(X)$ while decreasing $V_\mathrm{ud}(X)$ \textcolor{black}{because the time derivative of $V_\mathrm{ud}(X)$ remains nonpositive by \eqref{gibar:c} as follows:
\begin{align}\label{Vsign}
&\dot{V}_\mathrm{ud}(X(t)) 
= \sum_{i=1}^n \left(\nabla_i V_\mathrm{ud}(X(t))\right)^\top u_i(t) \nonumber\\
&= -\sum_{i\in\CV_\mathrm{t}} \left( (\nabla_i V_\mathrm{ud}(X(t)))^\top \bar{g}_i(X(t)) + 
\kappa_i \nabla_i \|V_\mathrm{ud}(X(t))\|^2\right) \nonumber\\
&\quad\quad -\sum_{i\in\CN\setminus\CV_\mathrm{t}} \mu_i \nabla_i \|V_\mathrm{ud}(X(t))\|^2 \leq 0.
\end{align}
Regarding its rigorous discussion, see Appendix \ref{subsec:lem_eql}}.

\textcolor{black}{
An interpretation of the proposed method in \eqref{ctr:digraph} and the function $\bar{g}_i(X)$ in \eqref{gbar_explicit} is as follows.
When improving the distributed gradient-flow method in \eqref{gd} with $V_\mathrm{ud}(X)$ by exploiting unidirectional edges, one can come up with the following distributed control input:
\begin{align}
\label{rem_ctr}
\!\!\!\!
 f_i(X) = \begin{cases}
  -\lambda_i \nabla_i \bar{V}_\mathrm{ud}(X) 
  -\kappa_{i}\displaystyle\nabla_i V_\mathrm{ud}(X) , \!\!\!& i \in \mathcal{V}_\mathrm{t} \\
  -\mu_{i}\displaystyle\nabla_i V_\mathrm{ud}(X) , \!\!\!& i\in \mathcal{N}\setminus\mathcal{V}_\mathrm{t}.
 \end{cases}
\end{align}
This controller, however, fails to guarantee the nonincrease of $V_\mathrm{ud}(X)$ because of the extra term
$-\lambda_i \nabla_i \bar{V}_\mathrm{ud}(X)$.
% behaves like a disturbance.
In contrast, the proposed method ensures $\dot{V}_\mathrm{ud}(X)\leq 0$ in \eqref{Vsign} by adopting the function $\bar{g}_i(X)$ in \eqref{gbar_explicit} in place of $\lambda_i \nabla_i \bar{V}_\mathrm{ud}(X)$ in \eqref{rem_ctr} to adjust the direction of $\lambda_i \nabla_i \bar{V}_\mathrm{ud}(X)$ in a similar manner to the Gram-Schmidt orthogonalization as \eqref{gbar_b} and \eqref{gbar_c}.
% &\dot{V}_\mathrm{ud}(X(t)) \\
% &= -\sum_{i\in\CV_\mathrm{t}} \kappa_i \|\nabla_i V_\mathrm{ud}(X(t))\|^2
% +
% \sum_{i\in\CN\setminus\CV_\mathrm{t}} \mu_i \|\nabla_i V_\mathrm{ud}(X(t))\|^2
% \\
% &- \lambda_i \sum_{i\in\CV_\mathrm{t}} \nabla_i V_\mathrm{ud}(X(t))^\top \nabla_i \bar{V}_\mathrm{ud}(X(t)),
% \end{align*}
% where the sign of the second term is unknown.
% In contrast, the proposed method ensures \eqref{Vsign} by adopting the proposed function $\bar{g}_i(X)$ in \eqref{gbar_explicit} in place of $\lambda_i \nabla_i \bar{V}_\mathrm{ud}(X)$ in \eqref{rem_ctr}, which adjusts the direction of $\lambda_i \nabla_i \bar{V}_\mathrm{ud}(X)$ in the manner like the Gram-Schmidt orthogonalization (see \eqref{gbar_b} and \eqref{gbar_c}) such that $\dot{V}_\mathrm{ud}(X)\leq 0$ is preserved.
}
% Fig.\ \ref{fig:gbar} shows a sketch of $\bar{g}_i$ in \eqref{gbar_explicit}.
%  Orange arrows, pink ones, blue ones, and black ones represent $\bar{g}_i$, $\lambda_i \nabla_i \bar{V}_\mathrm{ud}$, $\mu_i\nabla_i V_\mathrm{ud}$, and $\hat{g}_i$, respectively.
% Besides, the blue area is a region satisfying \eqref{gibar:c}.
% Figs. \ref{fig:gbar}a, \ref{fig:gbar}b, \ref{fig:gbar}c correspond to the cases (a), (b), and (c) in \eqref{gbar_explicit}, respectively.
% As shown in Fig.\ \ref{fig:gbar}, $\bar{g}_i(X)$ is a function that provides the orthogonal projection of $\hat{g}_i(X)$, which is the average of $\lambda_i \nabla_i \bar{V}_\mathrm{ud}$ and $\mu_i\nabla_i V_\mathrm{ud}$, onto the region where \eqref{gibar:c} is satisfied. 

% For the system \eqref{dynamics} with \eqref{controller} and \eqref{ctr:digraph}, the following lemma holds, which states the convergence to equilibrium points.
% Here, 
The following lemma indicates that the global attractiveness of the equilibrium set $\Omega$ is achieved with the proposed controller, where 
% $\Omega=\nabla V_\mathrm{ud}^{-1}(0)\cap \nabla_{\CV_\mathrm{t}} \bar{V}_\mathrm{ud}^{-1}(0)$
\begin{align}\label{def:Omega}
    \Omega=\nabla V_\mathrm{ud}^{-1}(0)\cap \nabla_{\CV_\mathrm{t}} \bar{V}_\mathrm{ud}^{-1}(0)
\end{align}
with $\nabla_{\CV_\mathrm{t}}\bar{V}_\mathrm{ud}^{-1}(0) := \left\{X \in \mathbb{R}^{d\times n} :\nabla_i \bar{V}_\mathrm{ud}(X)=0,\:\forall i\in \CV_\mathrm{t}\right\}$
% \begin{align}
% \label{def:Omega_A}
% \!\!\!\!\!
% \nabla_{\CV_\mathrm{t}}\bar{V}_\mathrm{ud}^{-1}(0) :=& \left\{X \in \mathbb{R}^{d\times n} :\nabla_i \bar{V}_\mathrm{ud}(X)=0,\:\forall i\in \CV_\mathrm{t}\right\}.\!\!
% \end{align}
% In addition, if $V_\mathrm{ud}(X)\to\infty$ holds when $\|X\|\to\infty$, then $V_\mathrm{ud}(X)$ is said to be \textit{radially unbounded}.}
\begin{lemma}\label{lem:eql}
Let time-invariant directed graph $G$ satisfy Assumption \ref{assu:g}.
Assume that $\bar{V}_\mathrm{ud}$ and $V_\mathrm{ud}$ are nonnegative, continuously differentiable, and radially unbounded.
% , and their level sets $L_{\bar{V}_\mathrm{ud}}(\rho)$ and $L_{V_\mathrm{ud}}(\rho)$ are non-empty and bounded for any $\rho>0$.
Then, for the system \eqref{dynamics} with \eqref{controller} and \eqref{ctr:digraph},
the set $\Omega$ in \eqref{def:Omega} is an equilibrium set and globally attractive for positive $\kappa_i>0,\; i\in\CV_\mathrm{t}$.
\end{lemma}
\begin{proof}
\textcolor{black}{
From \eqref{ctr:digraph}, \eqref{gibar}, and \eqref{gibar:c}, if $X\in \Omega= \nabla V_\mathrm{ud}^{-1}(0)\cap \nabla_{\CV_\mathrm{t}}\bar{V}_\mathrm{ud}^{-1}(0)$ is satisfied, then we obtain $\nabla_i V_\mathrm{ud}(X)=\bar{g}_i(X)=0$ for all $i\in\CN$, which gives $f_i(X)=0$ for all $i\in\CN$. Thus, from \eqref{dynamics} and \eqref{controller}, each point in $\Omega$ is an equilibrium point.}
Regarding the global attractiveness of $\Omega$, see Appendix \ref{subsec:lem_eql}.
\end{proof}

\textcolor{black}{
The constant $\kappa_i$ can be chosen as $0$ if $\nabla V_\mathrm{ud}(X)$ satisfies the following condition.
Setting $\kappa_i=0$ enhances the contribution of the function $\bar{g}_i(X)$ in \eqref{ctr:digraph}.}
\begin{coro}\label{coro:eql}
Consider the same assumption and the system as Lemma \ref{lem:eql}.
If $\nabla_i V_\mathrm{ud} (X)=0,\forall i\in \CN\setminus\CV_\mathrm{t}\Leftrightarrow\nabla_i V_\mathrm{ud} (X)=0, \forall i\in \CV_{\mathrm{t}}$ holds, $\Omega$ is globally attractive for $\kappa_i\geq0,\; i\in\CV_\mathrm{t}$.
\end{coro}
\begin{proof}
See Appendix \ref{subsec:coro}.
\end{proof}
% \begin{figure}[t]
%  \centering
%  \includegraphics[width=\linewidth]{fig_gibar.pdf}
%  \caption{Sketch of the function $\bar{g}_i$ in \eqref{gbar_explicit}.}
%  \label{fig:gbar}
% \end{figure}
% \textcolor{black}{
% \begin{rem}
% % Although the function $\bar{g}_i(\cdot)$ in \eqref{gbar_explicit} is piecewise defined, $\bar{g}_i(X)$ is bounded as long as $\nabla_i \bar{V}_\mathrm{ud}(X)$ and $\nabla_i V_\mathrm{ud}(X)$ take finite values.
% \end{rem}
% }

\textcolor{black}{
\begin{rem}\label{rem:existance}
    Despite the discontinuity in \eqref{ctr:digraph}, the existence of the solution is guaranteed in the sense of differential inclusions \cite{filippov1988differential,smirnov2022introduction}.
    For details, see Appendix \ref{subsec:lem_eql}.
    % the proof of Lemma \ref{lem:eql}.
\end{rem}
}
\subsection{Solution to Problem 1}\label{sec:solution}

This subsection provides a solution to Problem 1 in Section \ref{sec:setting} with the controller in \eqref{ctr:digraph}.

The following theorem shows that the proposed controller in \eqref{ctr:digraph} is distributed with respect to $G$ under the gradient-distributedness of $V_\mathrm{ud}$ and $\bar{V}_\mathrm{ud}$ with respect to $G_\mathrm{ud}$ and $\bar{G}_\mathrm{ud}$, respectively.
\begin{theorem}\label{thm:distributed}
Let the directed graph $G$ satisfy Assumption \ref{assu:g}.
Consider continuously differentiable and nonnegative functions $V_\mathrm{ud},\,\bar{V}_\mathrm{ud}:\BRdn\to\mathbb{R}$ satisfying the conditions C$1$ and C$3$, respectively.
% Consider a function $V_\RA:\BRdn\to\mathbb{R}$ satisfying the conditions A$1$, A$2$ and a function $V_\mathrm{ud}:\BRdn\to\mathbb{R}$ satisfying the conditions B$1$, B$2$.
Then, the controller in \eqref{ctr:digraph} is a distributed controller with respect to $G$.
\end{theorem}
\begin{proof}
In Appendix \ref{subsec:thm_distributed}, we provide a proof for the controller \eqref{ctr:digraph} with the generalization in Remark \ref{rem:gen}.
See Appendix \ref{subsec:thm_distributed} for more details.
\end{proof}

\textcolor{black}{
The key convergence property of the proposed method, which guarantees the performance enhancement from the gradient-flow method with $V_\mathrm{ud}(X)$, directly follows from Lemma \ref{lem:eql} and the assumptions of C2 and C4.}
Combining the following theorem and Theorem \ref{thm:distributed}, we can obtain a solution to Problem 1.
% and combining Theorem \ref{thm:distributed} and Theorem \ref{thm:sol} gives a solution to Problem 1.
\begin{theorem}\label{thm:sol}
Consider the same assumption as Lemma \ref{lem:eql}.
% Let time-invariant directed graph $G$ satisfy Assumption \ref{assu:g}.
% Assume that there exists an invariant set $\Omega_S\,(\supsetneq \Omega=\Omega_\RA\cap\Omega_\RB)$ of the system \eqref{dynamics} with \eqref{controller} and \eqref{ctr:digraph} such that the directed graph $G(X)$ is fixed over $\Omega_S$.
Assume that the functions $V_\mathrm{ud}$ and $\bar{V}_\mathrm{ud}$ satisfy the conditions C2 and C4, respectively.
Then, the inclusion in \eqref{def:problem_eql} holds for $\Omega$ in \eqref{def:Omega}.
\end{theorem}
% \begin{proof}
% follows from Lemma \ref{lem:eql} and the assumptions of C2 and C4.
% \end{proof}

Moreover, the important feature of the gradient-flow method, local attractiveness of $V_\mathrm{ud}^{-1}(0)$, is preserved.
\begin{theorem}\label{thm:gene_locally_attractive}
    Consider the same assumption as Lemma \ref{lem:eql}.
    Assume that $V_\mathrm{ud}(X)$ is a real analytic function in an open set containing $V_\mathrm{ud}(0)^{-1}$.
    Then, $V_\mathrm{ud}(0)^{-1}$ is locally attractive.
\end{theorem}
\begin{proof}
    From the boundedness of the level set of $V_\mathrm{ud}(X)$, the system with \eqref{controller} and \eqref{ctr:digraph} is bounded because $\dot{V}_\mathrm{ud}(X)\leq 0$ holds almost everywhere.
    Then, by using {\L}ojasiewicz's inequality in Theorem 6.3.4 in \cite{krantz2002primer} and Lemma \ref{lem:eql}, this theorem can be proved in the same way as the gradient-flow method (see Theorem 5.6 in \cite{sakurama2021generalized}).
\end{proof}
\textcolor{black}{
\begin{rem}\label{rem:graph_condition}
No connectivity assumption is required in solving Problem 1.
However, to achieve the control objective \eqref{def:objective2},
graph $G$ needs to satisfy a graph topological condition, e.g., connectivity.
% Such a condition must be satisfied in practice.
\end{rem}}
\textcolor{black}{
\begin{rem}\label{rem:extension}
    The idea of the proposed method can be extended to a more general directed graph than Assumption \ref{assu:g}; a multi-layered architecture of graphs satisfying Assumption \ref{assu:g}.
   % , such as \cite{hu2019distributed,li2020layered}.
   % Note that we can say that Assumption \ref{assu:g} corresponds to the case of two layers.
    For such graphs, a multi-layered controller can be designed by replacing $\nabla_i \bar{V}_\mathrm{ud}(X)$ in \eqref{ctr:digraph}, \eqref{gibar}, and \eqref{gibar:c} with the proposed method in \eqref{ctr:digraph} for the upper layer.
    The global convergence to an equilibrium point can be proved in the same way as Lemma \ref{lem:eql}.
\end{rem}}
\textcolor{black}{
\begin{rem}\label{rem:directed_extension}
We can relax Assumption \ref{assu:g} by combining existing methods (e.g., \cite{oh2011distance,park2015stabilisation,pham2017distance,zhang2019control}) with the proposed method.
Assumption \ref{assu:g} is a sufficient condition to guarantee the non-increase of gradient-distributed function $V_\mathrm{ud}$ in \eqref{Vsign}.
Thus, as long as an objective function $V_\mathrm{ud}$ satisfies \eqref{Vsign}, we can use a directed graph as $G_\mathrm{ud}$, which can alleviate Assumption \ref{assu:g}.
\end{rem}
}
% \begin{rem}
%  We 
% cannot compare these other methods with the proposed method in a fair way.
% The gradient-flow method is the only method that can be applied to general tasks and relies solely on relative positional information, similarly to the proposed method, to the authors' knowledge.
% % The gradient-flow method can be applied to directed graphs by ignoring unidirectional edges.
% In this sense, although \cite{lin2015graph,li2018dynamic,hu2019distributed} deal with relatively wide ranges of directed graph, they are not suitable for the comparison because they are dedicated to limited tasks, e.g., formation control.
% Other versatile approaches (e.g., distributed optimization and distributed model predictive control) require rich information exchange over positional measurements including global positions.
% % while the proposed method works even in situations where agents can obtain only positional measurements through sensing and have no other means of communication.
% % Thus we cannot compare these other methods with the proposed method in a fair way.
% \end{rem}
\textcolor{black}{
\begin{rem}\label{rem:rate}
Showing a convergence rate of the proposed method in \eqref{ctr:digraph} is difficult in general
because $\nabla V_\mathrm{ud}(X(t))=0 \Rightarrow u_i(t)=0$ for all $i\in\mathcal{N}$ is not satisfied in the proposed method in \eqref{ctr:digraph} by the additional term $\bar{g}_i(X)$ in \eqref{ctr:digraph}, differently from the gradient-flow method, $u_i = - \nabla_i V_\mathrm{ud}(X)$.
However, exponential convergence can be guaranteed for concrete cases as follows.
In distance-based formation control in Section \ref{sec:num}, the local convergence rate in Theorem \ref{thm:gene_locally_attractive} is exponential. This follows from the proof of Theorem 2 in \cite{sun2016exponential} and {\L}ojasiewicz's inequality \cite{sakurama2021generalized,krantz2002primer}.
In the dynamic matching problem in Section \ref{sec:application}, the local convergence to the desired target state in Theorem \ref{thm:ctr_matching_attractive} is proved to be exponential.
\end{rem}
}

\begin{rem}\label{rem:gen}

For a more flexible implementation of the proposed controller \eqref{ctr:digraph} and preparation for Section \ref{sec:application}, we present a generalization of the controller as follows.
Consider node sets $\CNRA$ and $\CNRB$ defined as
\begin{align}
\label{def:na}
&\CN_\mathrm{A}= \CV_{\mathrm{t}}\cup \hat{\CV}\\
\label{def:nb}
&\CN_\mathrm{B}=\CN\setminus(\CV_{\mathrm{t}}\cup \hat{\CV})
\end{align}
with an arbitrary subset $\hat{\CV}$ of $\CN\setminus(\CV_\mathrm{t}\cap\CV_\mathrm{h})$.
(Note that $\hat{\CV}$ is optional and can be empty.)
Then, by replacing $\CV_\mathrm{t}$ and $\CV\setminus\CV_\mathrm{t}$ in the proposed controller in \eqref{ctr:digraph} with $\CNRA$ and $\CNRB$,
we can use generalized edge sets $\CE_\RA$ and $\CE_\RB$ instead of $\bar{\CE}_\mathrm{ud}$ and $\CE_\mathrm{ud}$ in Problem 1 and the conditions C$1$--C$4$, respectively.
Here, $\CE_\RA$ and $\CE_\RB$ are edge sets
that satisfy the following conditions, respectively:
\begin{align}
\label{condition:ea}
&\bar{\CE}_\mathrm{ud}\subset \CE_{\mathrm{A}}\subset
\bar{\CE}_\mathrm{ud}\cup\{(i,j):\forall i,j\in\CN\setminus\CV_\mathrm{t},\, i\neq j\}\\
\label{condition:eb}
&\CE_{\mathrm{B}} \subset \CE_{\mathrm{ud}}.
\end{align}
For the controller with $\CNRA$, $\CNRB$, $\CE_\RA$, and $\CE_\RB$, we can similarly prove the distributedness and the same convergence as the controller in \eqref{ctr:digraph}.
In Section \ref{sec:application}, we apply the controller in \eqref{ctr:digraph} to a directed proximity graph with this generalization.
\end{rem}

% \textcolor{black}{
% \subsection{Beyond directed graphs in Assumption \ref{assu:g}}
% In this subsection, an approach to apply the core idea of the proposed method \eqref{ctr:digraph} to more general directed graphs beyond Assumption \ref{assu:g} is presented.
% Consider a layered graph $G$ and a target set which can be represented as 
% \begin{equation*}
%     \CT = \cap_{k}\{ X:\BRdn: [x_j]_{j\in \CN(G_k)\cup\CN(G_{k+1})}\in\CT_k \},
% \end{equation*}
% where $G_k$ represenets the $k$th layer of $G$.
% The directed graphs in Assumption \ref{assu:g} can be seen as a graph with two layers.
% For this type of tasks and systems, we can construct a distributed controller by "stacking" the proposed controller in \eqref{ctr:digraph} as follows:
% % \begin{}
% }
\subsection{Numerical Experiments}\label{sec:num}
\begin{figure*}[t]
 \centering
 \includegraphics[width=0.8\linewidth]{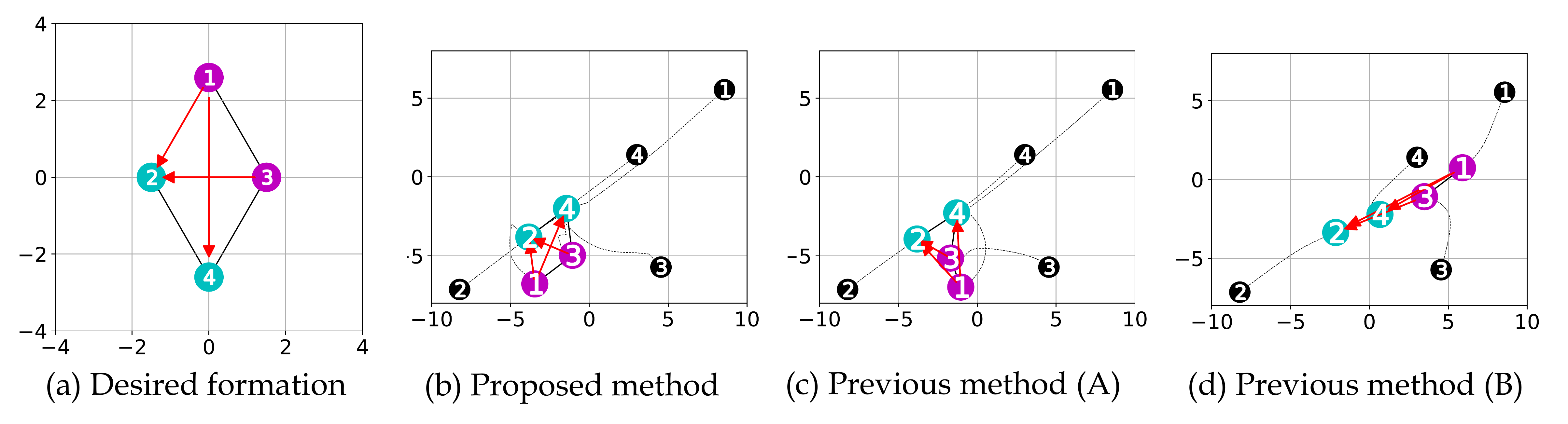}
 \caption{
 \textcolor{red}{Simulation results of distance-based formation with $4$ agents:
 (a) desired formation and the network topology;
 (b)-(d) the results of the proposed method in \eqref{ctr:digraph}, the previous method (A) in \eqref{eq:dist_ud}, and the previous method (B) in \eqref{eq:dist_di}.
Here, the black lines and the red arrows are undirected and directed edges, and the pink numbered circles and the blue ones are the agents in $V_\mathrm{t}$ and $\CN\setminus\CV_\mathrm{t}$.
Black numbered circles are the initial states, and dotted lines are the trajectories.}
 % An example of the $\delta$-proximity graph.
 }
 \label{fig:formation_4nodes}
\end{figure*}
\begin{figure*}[t]
 \centering
 \includegraphics[width=0.75\linewidth]{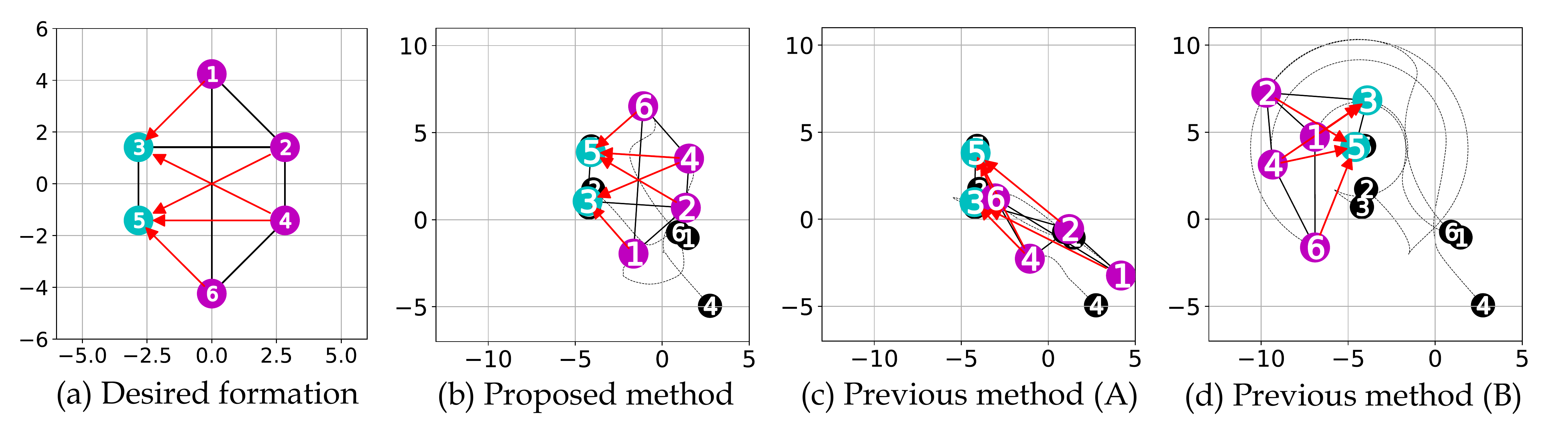}
 \caption{Simulation results of distance-based formation with $6$ agents.
%  (a)\&(e) desired formation and the network topology;
%  (b)-(d)\&(f)-(h) the results of the proposed method in \eqref{ctr:digraph}, the previous method (A) in \eqref{eq:dist_ud}, and the previous method (B) in \eqref{eq:dist_di}.
% Here, black lines and red arrows are undirected edges and directed ones, and pink numbered circles and blue ones are agents in $N_\RA$ and $N_\RB$.
% Black numbered circles are the initial states, and dotted lines are the trajectories.
 }
 \label{fig:formation_6nodes}
\end{figure*}
\begin{figure}[t]
 \centering
 \includegraphics[width=\columnwidth]{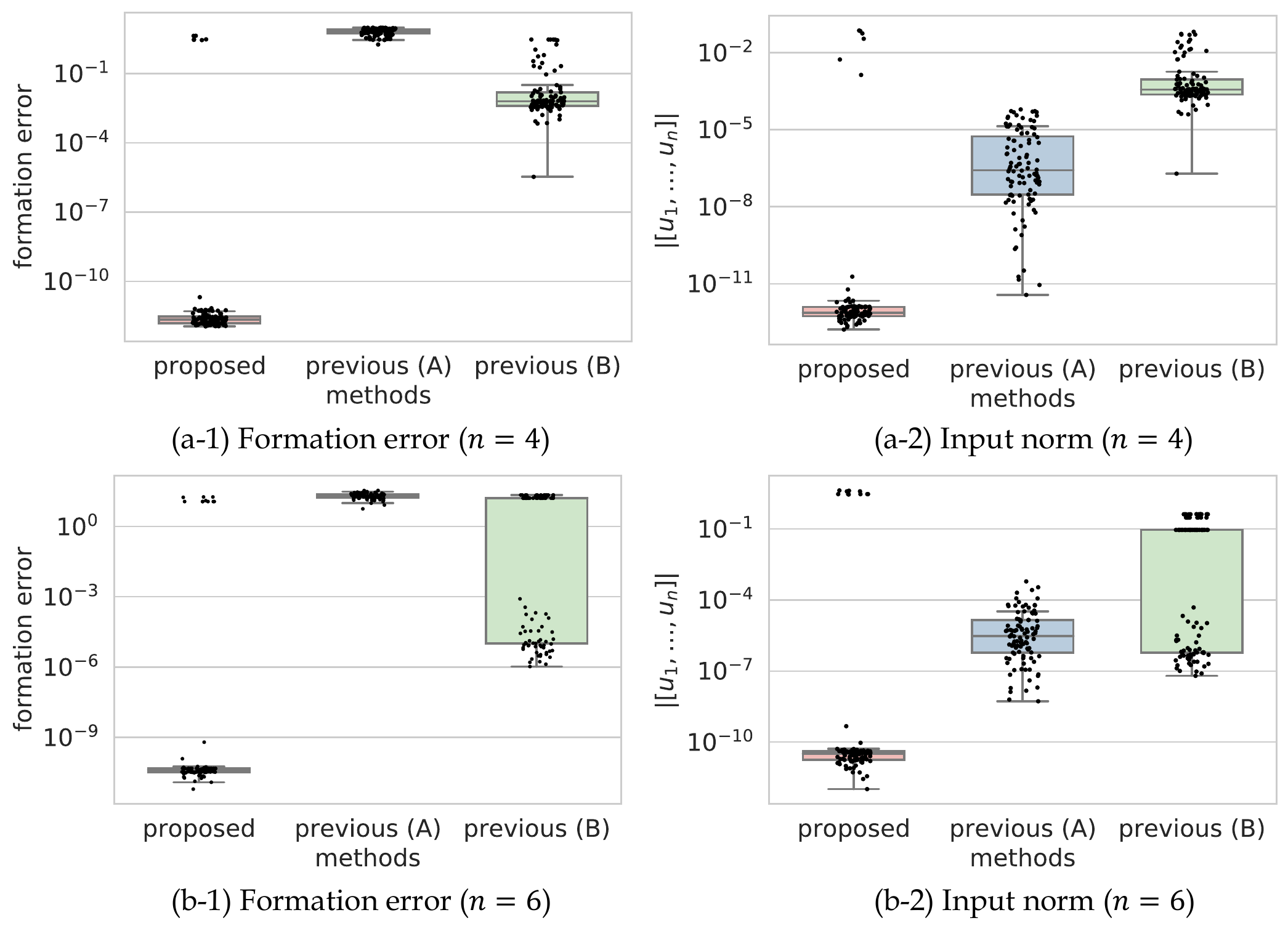}
 \caption{
 \textcolor{black}{
 \textcolor{red}{Comparison of the box plots of the formation error and input norm at $t=80$s for the proposed method and the previous methods (A) and (B).
 The upper and lower plots correspond to the cases of $n=4$ and $n=6$, respectively.
 Scattered dots represent the distribution map of the formation error and the input norm.}}
 }
 \label{fig:formation_hakohige}
\end{figure}

As a numerical study, we apply the proposed method to a distance-based formation control.
Consider systems with $4$ and $6$ agents in $2$D.
The dynamics of each agent is given as \eqref{dynamics}.
The desired configurations and fixed networks are given in Figs. \ref{fig:formation_4nodes}a and \ref{fig:formation_6nodes}a, respectively.
Here, the black lines and red arrows represent undirected and directed edges, respectively.
The desired distances between agents $i,j\in\CN,\,i\neq j$ are given by $d_{ij}=d_{ji}>0$.
Let the objective functions for $G_\mathrm{ud}$ and $\bar{G}_\mathrm{ud}$ be
\begin{align*}
    % &V_\mathrm{ud}(X) = \frac{1}{2} \sum_{(i,j)\in\CE_{\mathrm{ud}}} \phi_{ij}(x_i,x_j)\\
    % &\bar{V}_\mathrm{ud}(X) = \frac{1}{2} \sum_{(i,j)\in\bar{\CE}_\mathrm{ud}} \phi_{ij}(x_i,x_j)
    V_\mathrm{ud}(X) = \sum_{(i,j)\in\CE_{\mathrm{ud}}} \phi_{ij}(x_i,x_j),\;
    \bar{V}_\mathrm{ud}(X) = \sum_{(i,j)\in\bar{\CE}_\mathrm{ud}} \phi_{ij}(x_i,x_j)
    % \frac{1}{8} \sum_{(i,j)\in\CE_{\RB}} ( \|x_i-x_j\|^2-d_{ij}^2)^2
\end{align*}
with $\phi_{ij}(x_i,x_j)=(\|x_i-x_j\|^2-d_{ij}^2)^2/8$, which are the well-known conventional objective functions \cite{oh2015survey}.
In this simulation, we compared the proposed method in \eqref{ctr:digraph} with two previous methods: (A) the gradient-based controller in \eqref{gd} with $V=V_\mathrm{ud}$:
\begin{equation}\label{eq:dist_ud}
    f_i(x_i,[x_j]_{j\in\CN_i(G_\mathrm{ud})}) = - \nu_i \nabla_i V_\mathrm{ud}(X)
    % \frac{\nu_i}{2}\sum_{j\in\CN_i(G_\RB)} \nabla_i \phi_{ij}(x_i,x_j)
    % ( \|x_i-x_j\|^2-d_{ij}^2)(x_i-x_j),
\end{equation}
and (B) the distributed controller in the following form:
\begin{align}\label{eq:dist_di}
    &f_i(x_i,[x_j]_{j\in\CN_i(G)}) \nonumber \\ 
    =&- \nu_i \biggl(  \nabla_i V_\mathrm{ud}(X)  
    +\nabla_i \bigl( \sum_{
    (i,j)\in \CE_\mathrm{di}
    }\phi_{ij}(x_i,x_j) \bigr) 
    \biggr)
    % \nabla_i \hat{V}(X).
    % \frac{\nu_i}{2}\sum_{j\in\CN_i(G)} \nabla_i \phi_{ij}(x_i,x_j),
    % ( \|x_i-x_j\|^2-d_{ij}^2)(x_i-x_j).
\end{align}
Note that the controller in \eqref{eq:dist_di} has no guarantee of convergence because the opposite-directional function $\phi_{ji}(x_i,x_j)$ is not available \cite{schultz1962variable}.
Conversely, in the proposed method and the previous method (A), the agents are expected to converge to an equilibrium point from Lemma \ref{lem:eql}.
Furthermore, the proposed method is more likely to achieve the desired formation than the previous method (A) by Theorem \ref{thm:global}.
For the system with $4$ agents, $\CV_\mathrm{t}=\{1,3\}$ (pink nodes in Fig.\ \ref{fig:formation_4nodes}a) and $\CN\setminus\CV_\mathrm{t}=\{2,4\}$ (blue nodes in Fig.\ \ref{fig:formation_4nodes}a) hold.
The parameters in \eqref{gd}, \eqref{ctr:digraph}, and \eqref{eq:dist_di} were given as $\lambda_i=1.0,\,\mu_{i}=0.01,\,\kappa_i=0.01$, and 
$\nu_i = (\lambda_i+\mu_i)/2$.
For the system with $6$ agents, $\CV_\mathrm{t}=\{3,4,5\}$ (pink nodes in Fig.\ \ref{fig:formation_6nodes}a) and $\CN\setminus\CV_\mathrm{t}=\{1,2,6\}$ (blue nodes in Fig.\ \ref{fig:formation_6nodes}a) hold.
These parameters were given as $\lambda_i=2.0,\,\mu_i=0.01,\,\kappa_i=0.01$ and $\nu_i = (\lambda_i+\mu_i)/2$.

Figs.\ \ref{fig:formation_4nodes}b--d and \ref{fig:formation_6nodes}b--d present the simulation results of $n=4$ and $n=6$, respectively.
Figs. \ref{fig:formation_4nodes}b and \ref{fig:formation_6nodes}b present the results of the proposed method at $t=80$s.
Figs. \ref{fig:formation_4nodes}c and \ref{fig:formation_6nodes}c illustrate the results of the previous method (A) at $t=80$s, and Figs. \ref{fig:formation_4nodes}d and \ref{fig:formation_6nodes}d show those of the previous method (B) at $t=80$s, respectively.
Here, the dotted lines represent the trajectories of the agents, and the black circles correspond to the initial states.
As shown in Figs. \ref{fig:formation_4nodes}b and \ref{fig:formation_6nodes}b, by using the proposed method, the desired formations in Figs. \ref{fig:formation_4nodes}a and \ref{fig:formation_6nodes}a are finally achieved with some rotation and reflection.
Conversely, the agents do not converge to the desired configuration in the case of previous methods in Figs.\ \ref{fig:formation_4nodes}c--d and \ref{fig:formation_6nodes}c--d.
In particular, regarding Fig.\ \ref{fig:formation_6nodes}d, the agents do not converge to equilibrium points even in the final stage of the simulation.

For further comparison, we conducted simulations of the proposed and previous methods under the same parameters as the aforementioned simulations for $100$ different initial states.
These initial states are randomly generated by the uniform distribution for the interval $[-10,10]\times [-10,10]$.
Fig.\ \ref{fig:formation_hakohige} shows the box plots of the formation error $1/2 \,\sum_{i=1}^n\sum_{j=1}^n |\|x_i(t)-x_j(t)\|-d_{ij}|$ at $t=80$, and the norm $\|[[u_1(t),\ldots,u_n(t)]\|$ of the inputs at $t=80$.
% Figs. \ref{fig:formation_hakohige}a-1 and \ref{fig:formation_hakohige}a-2 plot the results for $n=4$, and
% Figs. \ref{fig:formation_hakohige}b-1 and \ref{fig:formation_hakohige}b-2 plot those for $n=6$.
% % The upper and lower plots correspond to $n=4$ and $n=6$, respectively.
% Here, the plots (a-1) and (b-1) in Fig. \ref{fig:formation_hakohige} represent the formation errors, and the plots (a-2) and (b-2) show the size $\|[u_1,\ldots,u_n]\|$ of the inputs at $t=80$.
% On the vertical axis is the formation error $1/2 \,\sum_{i=1}^n\sum_{j=1}^n |\|x_i-x_j\|-d_{ij}|$ at the end of the simulation, and on the horizontal axis is the employed method, i.e., the proposed method, the previous method (A), and the previous method (B).
Each black dot corresponds to the formation error and input norm for each initial state.
% The left plots are the results of the proposed method.
% Previous method (A) on the center and previous method (B) on the left side correspond to the gradient-flow controller in \eqref{gd} and the controller in \eqref{eq:dist_di}, respectively.
In most of the results shown in Fig.\ \ref{fig:formation_hakohige}a, the formation errors and input norms in the proposed method are significantly smaller than the others, which implies that the proposed method outperforms the others in both settings \textcolor{black}{thanks to the performance enhancement in Theorem \ref{thm:sol} and global attractiveness of $\CT$ in Theorem \ref{thm:gene_locally_attractive}.}
Note that the previous method (B) does not have any theoretical guarantee of convergence; thus, the continuous motions sometimes remain, as shown in Fig. \ref{fig:formation_6nodes}d.
% As for the case of $n=6$, the lower plot in Fig.\ \ref{fig:formation_hakohige} shows that the proposed method performed better than the others.
% In fact, the number of formation errors lower than $0.1$ was $56$ in the proposed method, eight in the previous method (A), and $26$ in the previous method (B).}

These results highlight the effectiveness of the proposed method.
\textcolor{black}{For further numerical results, see the preprint \cite{watanabe2023gradient_arxiv}.}

\section{Application to Distributed Dynamic Matching}\label{sec:application}

As an application to time-varying networks, we design a distributed controller for a new coordination task, dynamic matching, by adopting the design methodology in Section \ref{sec:digraph}.
\subsection{System Description}
\begin{figure}[t]
    \centering
    \includegraphics[width=0.42\columnwidth]{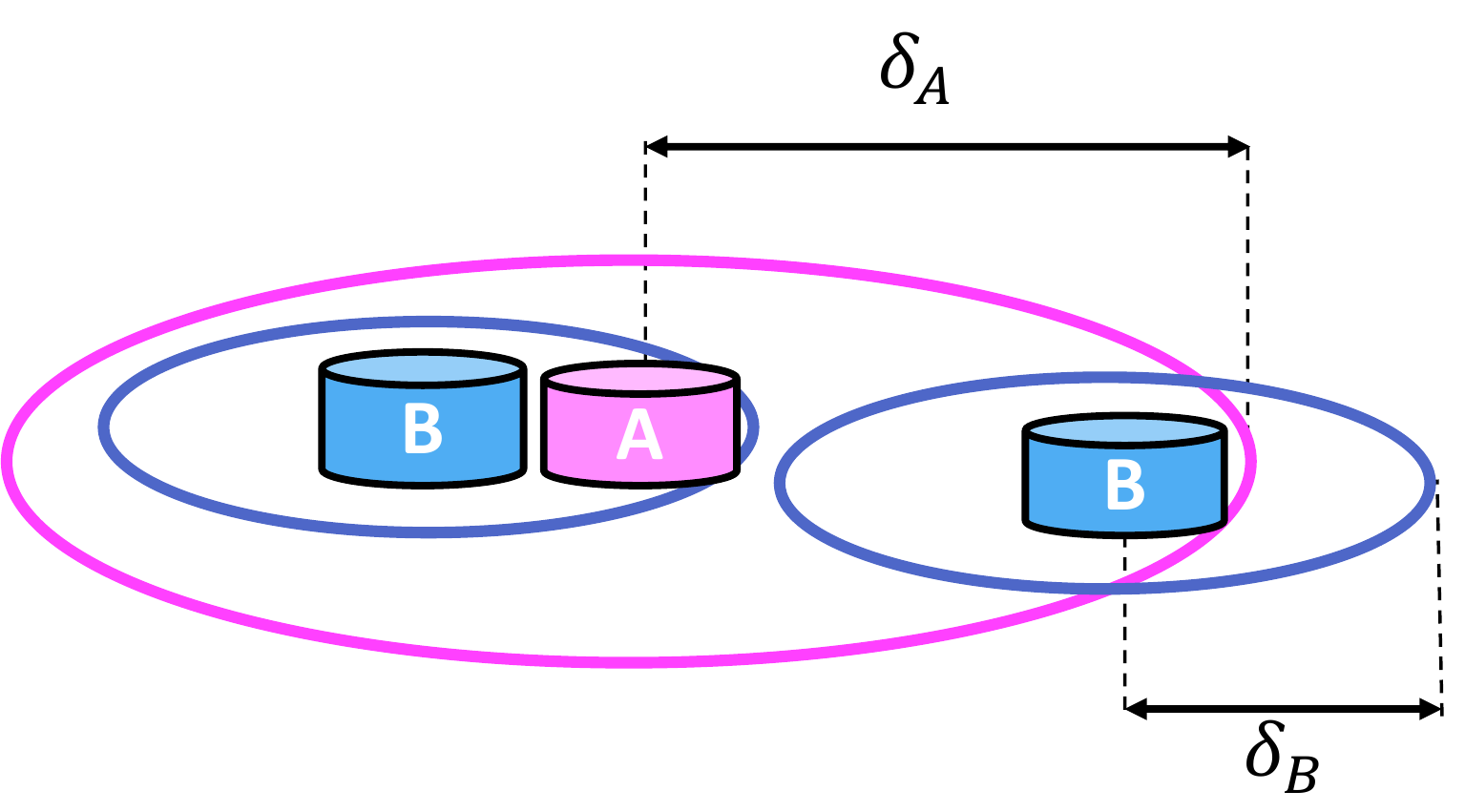} 
    \caption{Difference in sensing ranges between groups A and B. Agents in group A have a larger sensing range.}
    \label{fig:sensing}
\end{figure}

We assume that $n$ agents are classified into two groups, groups A and B. 
Our control objective is to achieve matching between the agents in groups A and B.
As shown in Fig.\ \ref{fig:sensing}, the agents have sensing ranges of different distances according to the groups within which agents can observe others.
% that the agents can observe an area at a specific distance called a \textit{sensing range}.
Let $\delta_\RA$ and $\delta_\RB$ be the distances of the sensing ranges of agents in groups A and B, respectively, satisfying $\delta_\RA>\delta_\RB>0$.
Let $n_\RA(<n)$ and $n-n_\RA$ be the numbers of agents in groups A and B, respectively.
Without loss of generality, we assume that agents from $1$ to $n_\RA$ belong to group A, and the others belong to group B.
Then, we define the index sets in groups A and B as $\CN_\RA=\{1,\ldots,n_\RA\}$ and $\CN_\RB=\{n_\RA+1,\ldots,n\}$, respectively. Note that the agent 
numbers of groups A and B are not necessarily equal.

We assume that agent $i\in\CNRA$ (resp. $\CNRB$) can obtain the information of the states of the agents within the distance $\delta_\RA$ (resp. $\delta_\RB$).
Accordingly, for $X=[x_1,\dots,x_n]$, the sensing topology of the agents is described by the state-dependent directed graph $G(X)=(\CN,\CE(X))$ with
\begin{align}\label{def:e}
     \mathcal{E}&(X)=\{(i,j):\|x_i-x_j\|\leq\delta_\RA,\,i\in\CNRA,\,j\in\CN,\,i\neq j \} \nonumber\\
      \quad&\cup\{(i,j):\|x_i-x_j\|\leq\delta_\RB,\,i\in\CNRB,\,j\in\CN,\,i\neq j \}.
\end{align}
Note that this graph is a directed proximity graph with two different sensing ranges.
 For example, if $\delta_\RB<\|x_i-x_j\|\leq \delta_\RA$ holds for $i\in\CNRA$ and $j\in\CNRB$, then agent $i\in\CNRA$ can observe agent $j\in\CNRB$, but not vice versa.

\subsection{Problem Formulation of Dynamic Matching}\label{subsec:application_problem}

Next, we formulate the dynamic matching problem.
The control objective is to positionally match the agents in the minority set $\Nmn$ with agents in the majority set $\Nmj$.
Here, we define the \textit{minority} and \textit{majority} sets, $\Nmn,\,\Nmj$, respectively, from the numbers of the elements in $\CN_\RA,\,\CN_\RB$ as follows:
\begin{align}
%\small
    \begin{cases}
 \Nmn=\CN_\mathrm{A},\:\Nmj=\CN_\mathrm{B}
 , &\mathrm{if}\;|\CN_\mathrm{A}|<|\CN_\mathrm{B}|\\
  \Nmn=\CN_\mathrm{B},\:\Nmj=\CN_\mathrm{A}
  , &\mathrm{if}\;|\CN_\mathrm{A}|>|\CN_\mathrm{B}|.
\end{cases}\nonumber
\end{align}
If $|\CN_\mathrm{A}|=|\CN_\mathrm{B}|$, we can assign either $\CN_\mathrm{A}$ or $\CN_\mathrm{B}$ to $\Nmn$ and $\Nmj$, provided $\Nmn\neq \Nmj$.
% Here, we say that the system achieves \textit{matching} when all agents in the minority set $\Nmn$ find their pairs from the majority set $\Nmj$ and the states of the paired agents converge to the same value.
The control objective can be formulated as 
\begin{align}
    \label{def:objective1}
    &\exists\alpha \in\Pi (\mathcal{N}_{\mathrm{mn}},\mathcal{N}_{\mathrm{mj}}) \nonumber\\
    &\mathrm{s.t.}\: \lim_{t \to \infty}(x_i(t)-x_{\alpha(i)}(t)) =0 \quad \forall i\in \mathcal{N}_{\mathrm{mn}},
\end{align}
where the pairs are determined by a one-to-one function $\alpha\in\Pi(\Nmn,\Nmj)$.
% The third line of \eqref{def:objective1} is introduced to prohibit agents in the same group from overlapping.
The expression in \eqref{def:objective1} can be rewritten as \eqref{def:objective2} with the target set given as follows:
\begin{align}
\label{def:target_set}
     \mathcal{T}=\bigcup_{\alpha\in \Pi(\mathcal{N}_\mathrm{mn},\mathcal{N}_\mathrm{mj})}&
     \!\!\!\{X\in \mathbb{R}^{d\times n}: 
    x_i = x_{\alpha(i)} \; \forall i\in\Nmn
    % , \nonumber \\
    % &x_i\neq x_j\;\forall i,j\in \Nmn,\, i\neq j
    \}.
\end{align}
\begin{figure}[t]
    \centering
    \includegraphics[width=0.68\columnwidth]{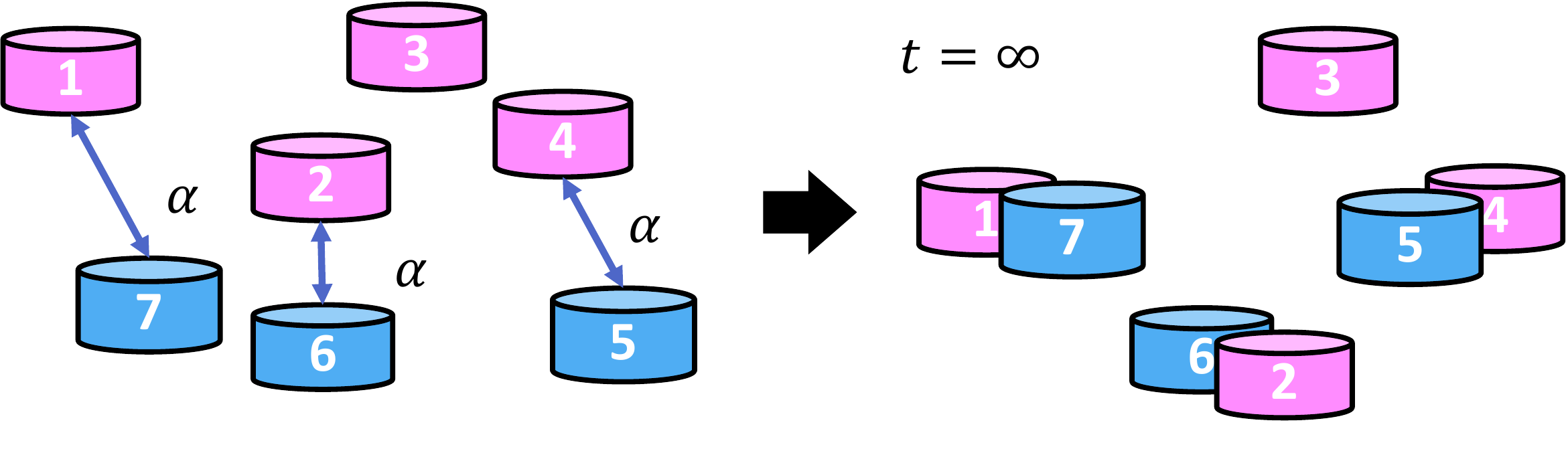} 
    \caption{Sketch of the control objective in \eqref{def:objective1}. 
Here, if  \eqref{def:objective1} is achieved, each agent in group B (minority group) finds its partner, and each pair converges to the same point.}
    \label{fig:matching_ex}
\end{figure}
\begin{exam}
Fig.\ \ref{fig:matching_ex} is a sketch of the situation in \eqref{def:objective1} for $\CN=\{1,2,\ldots,7\}$, $\CNRA=\{1,2,3,4\}$, and $\CNRB=\{5,6,7\}$.
Then, we have $\Nmn=\CNRB$ and $\Nmj=\CNRA$.
The one-to-one function $\alpha$ satisfies $\alpha(7)=1,\,\alpha(6)=2$, and $\alpha(5)=4$, and it is expected that the states of each pair converge to the same value
according to $\alpha$.
% As shown in Fig.\ \ref{fig:matching_ex}, $\alpha\in\Pi(\Nmn,\Nmj)$ takes its values as $\alpha(7)=1,\,\alpha(6)=2,\,\alpha(5)=4$ and the states of these paired agents converge to the same value without overlapping.
\end{exam}

\subsection{Preliminaries for Controller Design}\label{subsec:application_preliminary}

As a preliminary, we first prove that $G(X)$ with \eqref{def:e} satisfies Assumption \ref{assu:g} for each $X\in\BRdn$.
Here, from \eqref{def:e}, we obtain the following set of the unidirectional edges for each $X=[x_1,\ldots,x_n]\in\BRdn$:
\begin{align}
\label{e_di}
    \!\!\!\!
    \CE_\mathrm{di}=\{(i,j):\delta_\RB<\|x_i-x_j\|\leq \delta_\RA,\,i\in\CNRA,\,j\in\CNRB \}.
\end{align}
\begin{lemma}\label{lem:G}
Graph $G=G(X)$ in \eqref{def:e} satisfies Assumption \ref{assu:g} for any $X\in\BRdn$. 
\end{lemma}
\begin{proof}
From \eqref{def:v1}, \eqref{def:v2}, and \eqref{e_di}, $\CV_{\mathrm{t}}\subset\CNRA$ and $\CV_{\mathrm{h}}\subset\CNRB$ hold, which yields $\CV_{\mathrm{t}}\cap\CV_{\mathrm{h}}=\emptyset$ for any $X$. Therefore, $G(X)$ in \eqref{def:e} always satisfies Assumption \ref{assu:g}.
\end{proof}

Next, the generalization in Remark \ref{rem:gen} is valid in this problem
as follows.
% \textcolor{black}{Note that the proof is omitted due to the space limit. See the preprint \cite{watanabe2023gradient_arxiv} for the details.}
\begin{lemma}\label{lem:eaeb}
% For $G(X)$ in \eqref{def:e}, $\CNRA=\{1,\ldots,n_\RA\}$ and $\CNRB=\{n_\RA+1,\ldots,n\}$ satisfy \eqref{def:na} and \eqref{def:nb} for any $X\in\BRdn$.
For $G=G(X)$, $\CV_\mathrm{t}$ in \eqref{def:v1}, and $\CV_\mathrm{h}$ in \eqref{def:v2}, there exists a subset $\hat{\CV}$ of $\CN\setminus(\CV_{\mathrm{t}}\cup\CV_{\mathrm{h}})$ such that \eqref{def:na} and \eqref{def:nb} are satisfied for any $X\in\BRdn$.
% For $G=G(X)$ in \eqref{def:e} and \textcolor{black}{each} $X\in\BRdn$, there exists $\hat{\CV}\subset \bar{\CV}=\CN\setminus(\CV_{\mathrm{t}}\cup\CV_{\mathrm{h}})$ in \eqref{def:v1} and \eqref{def:v2} such that $\CNRA=\{1,\ldots,n_\RA\}$ and $\CNRB=\{n_\RA+1,\ldots,n\}$ satisfy \eqref{def:na} and \eqref{def:nb}.
\end{lemma}
\begin{proof}
    \textcolor{red}{For $\CV_{\mathrm{ud}}=\CN\setminus(\CV_{\mathrm{t}}\cup\CV_{\mathrm{h}})$, we choose $\hat{\CV}=\CV_{\mathrm{ud}}\cap\CNRA$.
Then, from \eqref{def:v1}, \eqref{def:v2}, \eqref{e_di}, and De Morgan's laws, the following two inclusion relationships hold:
\begin{align}
\label{proof:na}
&\CV_\mathrm{t}\cup\hat{\CV} \subset \CNRA \cup( (\CN\setminus(\CV_\mathrm{t}\cup\CV_\mathrm{h}))
    \cap \CNRA )= \CNRA,\\
    \label{proof:nb}
&\CV_\mathrm{h}\cup(\CV_{\mathrm{ud}}\setminus\hat{\CV})\nonumber\\     
    &\quad\subset \CNRB \cup ( (\CN\setminus(\CV_\mathrm{t}\cup\CV_\mathrm{h})) 
     \setminus ( (\CN\setminus(\CV_\mathrm{t}\cup\CV_\mathrm{h})) \cap \CNRA ) ) \nonumber \\
    &\quad=\CNRB \cup ((\CN\setminus(\CV_\mathrm{t}\cup\CV_\mathrm{h}))\cap((\CV_\mathrm{t}\cup\CV_\mathrm{h})\cup\CNRB)) \nonumber\\
    &\quad\subset \CNRB \cup\CNRB = \CNRB.
\end{align}
Moreover, $(\CV_\mathrm{t}\cup\hat{\CV})\cup( \CV_\mathrm{h}\cup(\CV_{\mathrm{ud}}\setminus\hat{\CV}))=\CN$ and $(\CV_\mathrm{t}\cup\hat{\CV})\cap( \CV_\mathrm{h}\cup(\CV_{\mathrm{ud}}\setminus\hat{\CV}))=\emptyset$ hold.
Hence, from \eqref{proof:na} and \eqref{proof:nb}, we obtain $\CV_\mathrm{t}\cup\hat{\CV}=\CNRA$ and $\CV_\mathrm{h}\cup(\CV_{\mathrm{ud}}\setminus\hat{\CV})=\CNRB$; thus \eqref{def:na} and \eqref{def:nb} are satisfied.}
\end{proof}

In addition, for each $X\in\BRdn$ and $G(X)$, the undirected proximity graphs $G_{\delta_\RA}(X)$ and $G_{\delta_\RB}(X)$ can be assigned to the generalized networks in Remark \ref{rem:gen}.
\textcolor{black}{The proof is also omitted due to the space limit. See the preprint \cite{watanabe2023gradient_arxiv} for details.}
\begin{lemma}\label{lem:gagb}
For $\CE_\RA=\CE_{\delta_\RA}(X)$ and $\CE_\RB=\CE_{\delta_\RB}(X)$, the conditions in \eqref{condition:ea} and \eqref{condition:eb} are satisfied for any $ X\in\BRdn$.
\end{lemma}
\begin{proof}
    \textcolor{red}{First, consider $G_\mathrm{ud}=G_{\delta_\RB}(X)$. From \eqref{def:e_delta} and \eqref{def:e}, $\CE_{\delta_\RB}(X)\subset\CE(X)$ holds. The edges of $\CE_{\delta_\RB}(X)$ are bidirectional; thus \eqref{condition:eb} holds. 
Next, we consider $\CE_\RA=\CE_{\delta_\RA}(X)$. From \eqref{def:e} and \eqref{e_di}, we obtain $\{(i,j):(j,i)\in\CE_\mathrm{di}\}=\{(i,j): 
 \delta_\mathrm{B} < \|x_i-x_j\| \leq\delta_\mathrm{A},\,i\in  \CN_\mathrm{B},\,j\in  \CN_\mathrm{A}\}$.
 Then, for $\bar{\CE}_\mathrm{ud}$ in \eqref{def:ebar},
 \begin{align*}
 %\small
 \bar{\CE}_\mathrm{ud}&=\{(i,j):\|x_i-x_j\|\leq\delta_\RA,\,i\in\CNRA,\,j\in\CN,\,i\neq j \} \\
      &\cup\{(i,j):\|x_i-x_j\|\leq\delta_\RB,\,i\in\CNRB,\,j\in\CN,\,i\neq j \}
 \\
 &\cup \{(i,j): 
 \delta_\mathrm{B} < \|x_i-x_j\| \leq\delta_\mathrm{A},\,i\in  \CN_\mathrm{B},\,j\in  \CN_\mathrm{A}\} \\
 &=\,\mathcal{E}_{\delta_\mathrm{A}}(X)\setminus\{(i,j):\delta_\mathrm{B} <\; \|x_i-x_j\|\leq\delta_\mathrm{A},\,  \\
&\quad\quad\quad\quad\quad\quad\quad\quad\quad\quad i,j\in \mathcal{N}_\mathrm{B},\,i\neq j\}\subset\CE_{\delta_\mathrm{A}}(X)
 \end{align*}
 is derived from \eqref{def:e}. Thus, we obtain $\mathcal{E}_{\delta_\mathrm{A}}(X)\subset\mathcal{E}_{\delta_\mathrm{A}}(X)\cup\{(i,j):\forall i,j\in\CN_\mathrm{B},i\neq j\}
      =\bar{\CE}_\mathrm{ud}\cup \{(i,j):\forall i,j\in\CN_\mathrm{B},i\neq j\}$.
 Hence, $\CE_\RA=\CE_{\delta_\RA}(X)$ satisfies the condition in \eqref{condition:ea}.}
\end{proof}

Finally, we construct objective functions for dynamic matching.
For $\CC\subset\CN$, let $\CC_\RK:=\CC\cap\CN_\RK$ for $\RK=\RA,\RB$, 
and we define $\Cmn$ and $\Cmj$ as follows:
\begin{align}
%%\small
\begin{cases}
\Cmn=\CC_\mathrm{A},\,\Cmj=\CC_\mathrm{B} &\mathrm{if}\;|\CC_\mathrm{A}|<|\CC_\mathrm{B}|\\
\Cmn=\CC_\mathrm{B},\,\Cmj=\CC_\mathrm{A} &\mathrm{if}\;|\CC_\mathrm{A}|>|\CC_\mathrm{B}|.
\end{cases}\nonumber
\end{align}
In other words, according to the number of agents in groups A and B within $\CC$, the minority and majority sets are assigned to $\Cmn$ and $\Cmj$, respectively.
If $|\CC_\RA|=|\CC_\RB|$, we can assign either $\CC_\RA$ or $\CC_\mathrm{B}$ to $\Cmn$ and $\Cmj$, as long as $\Cmn\neq \Cmj$.
Now, we consider the objective functions $\bar{V}_\mathrm{ud}(X)$ and $V_\mathrm{ud}(X)$ in the following form:
\begin{align}
\small
 \label{def:Vbar}
& \bar{V}_\mathrm{ud}(X) = \sum_{\CC\in \mathrm{clq}(\bar{G}_\mathrm{ud})}v_\CC([x_j]_{j\in\CC})\\
 \label{def:Vud}
 &V_\mathrm{ud}(X) = \sum_{\CC\in \mathrm{clq}(G_{\mathrm{ud}})}v_\CC([x_j]_{j\in\CC}),
\end{align}
where $\bar{G}_\mathrm{ud}=G_{\delta_\RA(X)}$ and $G_\mathrm{ud}=G_{\delta_\RB(X)}$.
The function $v_\CC$ for each maximal clique $\CC$ is
% for $\CC\in\mathrm{clq}(G_{\delta_\RA}(X))\cup\mathrm{clq}(G_{\delta_\RB}(X))$, 
given as follows:
\begin{align}\label{vc}
\small
v_\CC([x_j]_{j\in \mathcal{C}})
% &=\frac{1}{2}\min_{\alpha_{\mathcal{C}} \in \Pi(\mathcal{C}_{\mathrm{mn}},\mathcal{C}_{\mathrm{mj}})} 
%  \|[x_j]_{j\in \mathcal{C}_{\mathrm{mn}}}  - [x_{\alpha_{\mathcal{C}}(j)}]_{j\in \mathcal{C}_{\mathrm{mn}}}\|^2 \nonumber \\
                                % &
                                =\frac{1}{2}\,\|[x_j]_{j\in \mathcal{C}_{\mathrm{mn}}} 
 - [x_{\bar{\alpha}_{\mathcal{C}}(j)}]_{j\in \mathcal{C}_{\mathrm{mn}}}\|^2
\end{align}
with \begin{align}
\label{def:alphaC}
 \bar{\alpha}_{\mathcal{C}} \in \argmin_{\alpha_{\mathcal{C}} \in \Pi(\mathcal{C}_{\mathrm{mn}},\mathcal{C}_{\mathrm{mj}})} 
 \|[x_j]_{j\in \mathcal{C}_{\mathrm{mn}}} 
 - 
 [x_{\alpha_{\mathcal{C}}(j)}]_{j\in \mathcal{C}_{\mathrm{mn}}}\|.
\end{align}
Then, the functions $V_\mathrm{ud}$ in \eqref{def:Vud} and $\bar{V}_\mathrm{ud}$ in \eqref{def:Vbar} 
are gradient-distributed with regard to $G_{\delta_\RA}$ and $G_{\delta_\RB}$ from Theorem 1 in \cite{sakurama2014distributed}, and the inclusion $\CT\subset V_\RK^{-1}(0)$ holds for $\RK=\RA,\,\RB$.
Note that when $\Cmn=\emptyset$, then $v_\CC$ is identically zero.

For each $i\in\CN$, the gradient of $v_\CC$ for $\CC\in\mathrm{clq}_i (G_{\delta_\RK}(X))$ is reduced to 
\begin{equation}
\label{grad_vc}
    \nabla_i v_{\CC}([x_j]_{j\in \mathcal{C}}) = \begin{cases}
     x_i-x_{\bar{\alpha}_{\mathcal{C}}(i)}, &\!\!\! i\in \mathcal{C}_{\mathrm{mn}} \\
     x_i-x_{\bar{\alpha}^{-1}_{\mathcal{C}}(i)},
     &\!\!\! i\in \mathcal{C}_{\mathrm{mj}}\cap\bar{\alpha}_\CC(\mathcal{C}_\mathrm{mn}) \\
     0,&\!\!\! i\in \mathcal{C}_{\mathrm{mj}}\setminus\bar{\alpha}_\CC(\mathcal{C}_\mathrm{mn}).
     \end{cases}
\end{equation}
% with $\bar{\alpha}_\CC$ in \eqref{def:alphaC}.
% Note that we can define this derivative in a sufficiently small open neighborhood around $[x_j]_{j\in \mathcal{C}}$.
% % Here, $\partial_i$ denotes the operation of taiking subgradient with respect to $x_i$.
\subsection{Distributed Controller for Dynamic Matching}\label{sec:controller}

Based on the problem setting, preliminaries, proposed controller in \eqref{ctr:digraph}, and Remark \ref{rem:gen}, a distributed controller for dynamic matching is designed as follows:
\begin{align}
\label{ctr:matching}
%\small
 f_i(X) = \begin{cases}
  -\bar{g}_i(X) -\kappa_{i}\displaystyle \nabla_i V_\mathrm{ud}(X) , & i \in \CN_\mathrm{A} \\
  -\mu_i\displaystyle \nabla_i V_\mathrm{ud}(X) , & i\in \CN_\mathrm{B}
 \end{cases}
\end{align}
with $\bar{V}_\mathrm{ud}$ in \eqref{def:Vbar}, $V_\mathrm{ud}$ in \eqref{def:Vud}, and $\bar{g}_i$ in \eqref{gibar:c}.
This controller has been explicitly presented in Theorem 1 in \cite{watanabe2022matching}.
% Note that $V_\mathrm{ud}$ and $\bar{V}_\mathrm{ud}$ in this paper are written as $V_\mathrm{ud}$ and $V_\RA$, respectively, in \cite{watanabe2022matching}.

The designed controller can be implemented distributedly in the following procedure:
\begin{enumerate}
    \item Measure the neighboring agents' relative states $x_j-x_i$ for $j\in\CN_i(G(X))$.
    \item Compute all maximal cliques $\mathrm{clq}(G_{\delta_\RK}(X))$
    of the subgraph $(\CN_i(G_{\delta_\RK}(X))\cup\{i\}, \CE_{\RK i}(X))$ for $\RK=\RA,\,\RB$ with $\CE_{\RK i}(X)=\{(k,l): \|x_k-x_l\|\leq \delta_{\RK},\, k,l\in  \CN_i(G_{\delta_\RK}(X))\cup\{i\} \}$.
    \item Obtain a mapping $\bar{\alpha}_{\CC}$ using \eqref{def:alphaC} for each $\CC \in \mathrm{clq}(G_{\delta_\RK}(X)),\,\RK=\RA,\,\RB$.
    \item Compute $u_i(t)$ with \eqref{controller}, \eqref{def:Vbar}, \eqref{def:Vud}, \eqref{grad_vc}, and \eqref{ctr:matching}.
\end{enumerate}
In Step 2, each agent computes its belonging maximal cliques from the local subgraph
$(\CN_i(G_{\delta_\RK}(X))\cup\{i\}, \CE_{\RK i}(X))$ without using global information.
Note that the agents in $\CNRB$ do not have to consider $G_{\delta_\RA}(X)$ in Steps 2 and 3.

For this controller, we present the following convergence theorem.
This theorem shows that the target set $\CT$ for dynamic matching is locally exponentially attractive almost everywhere.
\begin{theorem}\label{thm:ctr_matching_attractive}
Consider the system in \eqref{dynamics} with a state feedback controller in \eqref{controller} and \eqref{ctr:matching} with \eqref{vc} for constants $\mu_i,\,\lambda_i,\,\eta_i>0,\, i\in \CN$, and $\kappa_i\geq 0,\, i\in\CNRA$.
Then, $\CT\setminus(\CU_1\cup\CU_2)$ is locally attractive, where
$\CU_1,$ and $\CU_2\subset\BRdn$ are defined as
\begin{align}
    \label{def:overlapping}
    \CU_1& = \{X\in\BRdn: x_i = x_j,\,\exists i,j\in\Nmj,\, i\neq j \}
\\
    \label{def:margin_delta_K}
    \CU_2& = \bigcup_{\RK=\RA,\RB}\{X\in\BRdn: \nonumber\\
    &\quad  \exists (i,j)\in\Nmn\times\Nmj\;\mathrm{s.t.}\;
    \|x_i-x_j\| = \delta_\RK
    \},
\end{align}
respectively.
\textcolor{black}{
Moreover, the convergence is exponential.
}
\end{theorem}
\begin{proof}
    See Appendix \ref{subsec:proof_thm_attraction}.
\end{proof}

Furthermore, if the sensing range is sufficiently large, the global attractiveness of $\CT$ is achieved as follows.
% $X(t)\to \CT$ is globally achieved as follows.
\begin{theorem}\label{thm:global}
Under the same setting as Theorem \ref{thm:ctr_matching_attractive}, if $\delta_\RB>0$ is sufficiently large, then $\CT$ is globally attractive for any $\lambda_{i},\,\mu_{i},\,\eta_i>0$, and $\kappa_i\geq 0$.
\end{theorem}
\begin{proof}
    % \subsection{Proof of Theorem \ref{thm:global}}\label{sec:proof_thm_global}
When $\delta_\RB>0$ is sufficiently large, $G(X)$ is a complete graph and $G_{\delta_\mathrm{A}}(X)=G_{\delta_\mathrm{B}}(X)=G(X)$ always holds.
Then, we have $\mathrm{clq}(G_{\delta_\mathrm{A}}(X))=\mathrm{clq}(G_{\delta_\mathrm{B}}(X) )=\{\mathcal{N}\}$.
Thus, from \eqref{def:Vud}, \eqref{def:Vbar}, \eqref{vc}, and \eqref{grad_vc}, we obtain
\begin{align*}
\small
\!\! \nabla_i V_\mathrm{ud} (X)
\!=\nabla_i \bar{V}_\mathrm{ud} (X)
\!=\!\!
 \begin{cases}
  x_i-x_{\bar{\alpha}_{\mathcal{N}}(i)}, \!\!\!\!& i \in \mathcal{N}_{\mathrm{mn}}\\
  x_i-x_{\bar{\alpha}_{\mathcal{N}}^{-1}(i)},  \!\!\!\!& i \in \mathcal{N}_{\mathrm{mj}}\cap \bar{\alpha}_{\mathcal{N}}(\mathcal{N}_{\mathrm{mn}})\\
  0,  \!\!\!\!&i\in \mathcal{N}_{\mathrm{mj}}\setminus \bar{\alpha}_{\mathcal{N}}(\mathcal{N}_{\mathrm{mn}})
 \end{cases}
\end{align*}
almost everywhere. Hence, the local and global minima of $V_\mathrm{ud}(X)$ and $\bar{V}_\mathrm{ud}(X)$ are the same.
Moreover, from \eqref{dynamics}, \eqref{controller}, and \eqref{ctr:digraph}, we obtain $\dot{V}_\mathrm{ud}(X(t))\leq 0$ almost everywhere. Thus, from Theorem 3.2 in \cite{shevitz1994lyapunov}, $X(t)\to \CT$ is globally achieved.
\end{proof}

\subsection{Numerical Experiment of Dynamic Matching}\label{subsec:num_matching}

\begin{figure}[t]
    \centering
    \includegraphics[width=0.70\columnwidth]{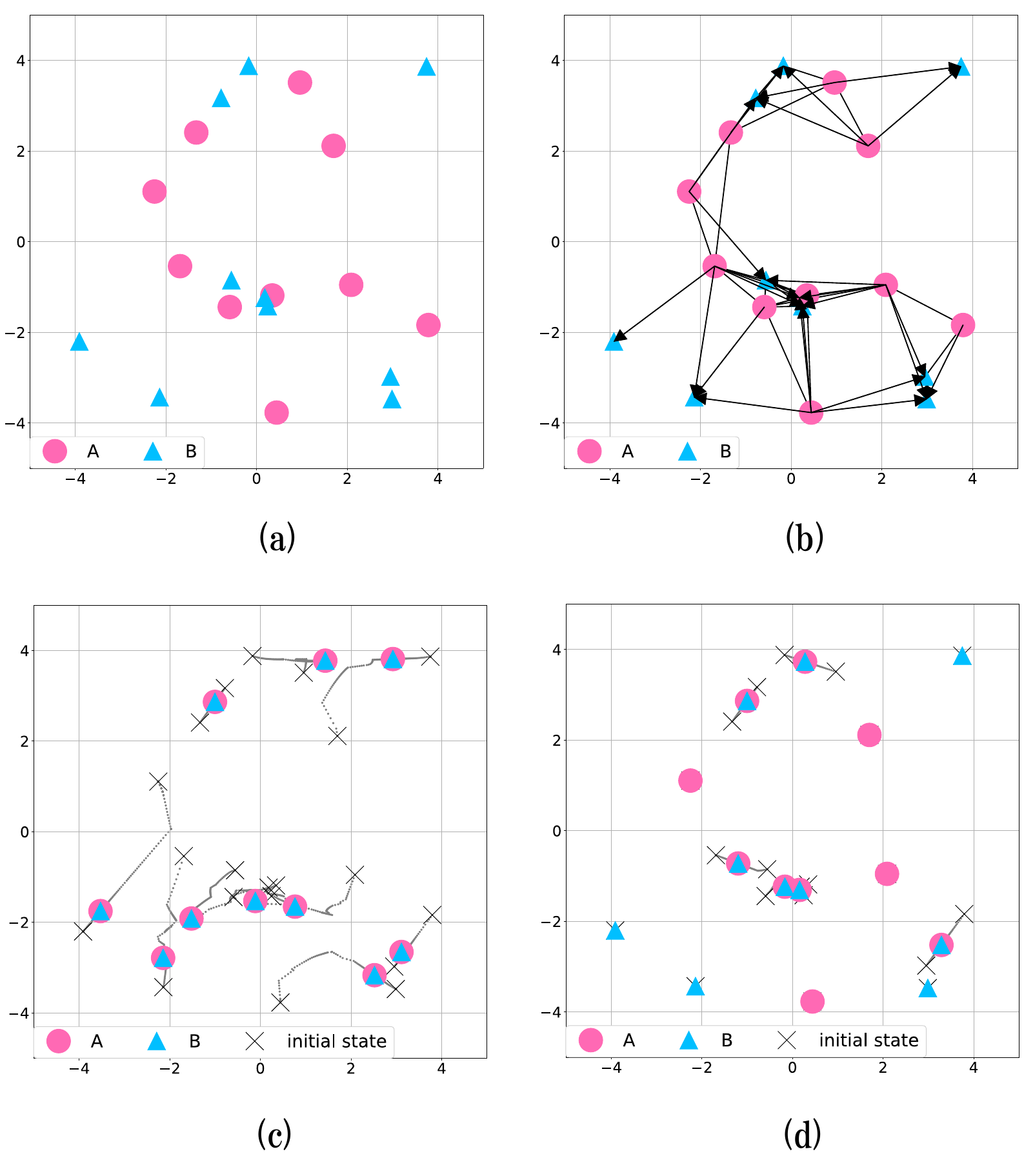} 
    \caption{(a) Initial states $X(0)$, (b) initial network $G(X(0))$, (c) simulation result of our proposed method in \eqref{ctr:digraph}, and (d) result of the previous method in \eqref{gd}. 
The pink circles and blue triangles represent the positions of agents in groups A and B, respectively.
    Dotted lines represent the trajectories.} \label{fig:sim}
\end{figure}

\begin{figure}[t]
    \centering
    \includegraphics[width=0.85\columnwidth]{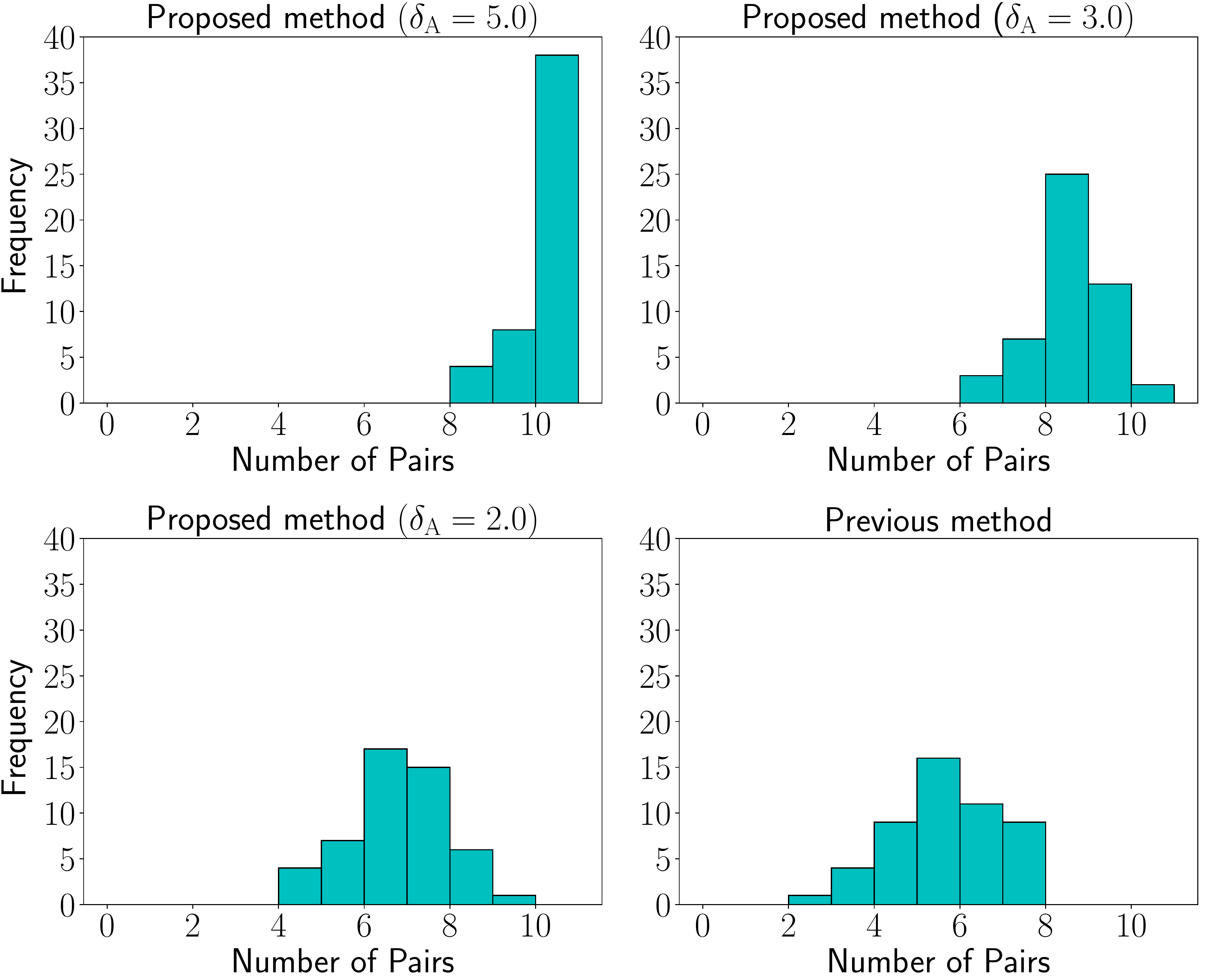} 
    \caption{
    \textcolor{red}{Histograms of the number of achieved pairs using the proposed method in \eqref{ctr:digraph} and the previous method in \eqref{gd} with 50 different initial states for $\delta_\RA=2.0,\,3.0,\,5.0$.}}
    % (a) Initial states $X(0)$, (b) initial network $G(X(0))$, (c) simulation result of our proposed method, and (d) result of the previous method. Dotted lines represent trajectories. }
    \label{fig:hist}
\end{figure}

% We demonstrate the effectiveness of our proposed method through a numerical example.
In the simulation, we consider $n=20$ agents in $2$-dimensional space.
Here, we set $n_\RA=n-n_\RA=10$.
% Let the number of agents in groups A and B be $n_\RA=n-n_\RA=10$.
The dynamics of the agents is expressed in \eqref{dynamics}, and the network of agents is expressed in \eqref{def:e} with $\delta_\RA=3.0$ and $\delta_\RB=1.5$.
We conduct a simulation of the proposed controller in \eqref{ctr:matching} with $\lambda_i=1.20,\,\mu_i=\eta_i=0.60,\, i\in \CN$ and $\kappa_i=0,\, i\in\CNRA$.
For comparison, the result of the previous method in \eqref{gd} with $V=V_\mathrm{ud}$ is also presented, where
% Here, from the proof of Theorem \ref{thm:ctr_matching_attractive} in Subsection \ref{subsec:proof_thm_attraction}, for the proposed method \eqref{ctr:matching_second}, the corresponding controller to \eqref{ctr:gradflow} is given as follows:
% \begin{equation}
% \label{ctr:gradientflow_matching}
%     f_i(X) = -\nu_i\Delta_{\RB i}(X),\;i\in\CN
% \end{equation}
the parameters are assigned as $\nu_i=0.90$ for $i\in\CNRA$, and $\nu_i=0.60$ for $i\in\CNRB$.
% because then \eqref{ctr:gradientflow_matching} become identical to our proposed one \eqref{ctr:matching_second} with $\delta_\RA=\delta_\RB$.

Fig.\ \ref{fig:sim} presents the simulation result. Fig.\ \ref{fig:sim}a illustrates the initial state $X(0)$, where the pink circles and blue triangles represent the positions of agents in groups A and B, respectively.
Fig.\ \ref{fig:sim}b shows graph $G(X(0))$.
Figs.\ \ref{fig:sim}c and \ref{fig:sim}d present the results of the proposed and previous methods, respectively, where the trajectories are drawn with dotted lines. 
In both cases, the agents converge to their equilibrium points.
Fig.\ \ref{fig:sim}c indicates that the proposed method achieves matching of all agents. On the other hand, only six pairs are obtained in the previous method from Fig. \ref{fig:sim}d.
This result demonstrates the effectiveness of the proposed method.

\textcolor{red}{
Furthermore, we conduct simulations under 50 different initial states $x_i(0)\in[-4,4]^2,\,i\in\CN$ for $\delta_\RA=2.0,\,3.0,\,5.0$.
The other conditions are the same as in Fig.\ \ref{fig:sim}.
Fig.\ \ref{fig:hist} plots the histograms of the number of achieved pairs by adopting the proposed and previous methods.
% with 50 different initial states for $\delta_\RA=2.0,\,3.0,\,5.0$. 
Fig.\ \ref{fig:hist} shows that the proposed method can achieve more pairs than the previous one, and the larger the value of $\delta_\RA$ is, the more pairs achieve matching. In particular, in the case of $\delta_\RA=5.0$, matching all agents was successful 38 times for 50 different initial states.
These results demonstrate the effectiveness of the proposed method.
Therefore, our proposed method's improved performance and its effectiveness are verified.}
% \textcolor{black}{For further numerical results, see the preprint \cite{watanabe2023gradient_arxiv}.}
% Furthermore, we conduct simulations under 50 different initial states $x_i(0)\in[-4,4]^2,\,i\in\CN$ for $\delta_\RA=2.0,\,3.0,\,5.0$.
% The other conditions are the same as in Fig.\ \ref{fig:sim}.
% Fig.\ \ref{fig:hist} plots the histograms of the number of achieved pairs by adopting the proposed and previous methods.
% % with 50 different initial states for $\delta_\RA=2.0,\,3.0,\,5.0$. 
% Fig.\ \ref{fig:hist} shows that the proposed method can achieve more pairs than the previous one, and the larger the value of $\delta_\RA$ is, the more pairs achieve matching. In particular, in the case of $\delta_\RA=5.0$, matching all agents was successful 38 times for 50 different initial states.
% These results demonstrate the effectiveness of the proposed method.
% Therefore, our proposed method's improved performance and its effectiveness are verified.
\section{Conclusions}\label{sec:conc}

In this study, we proposed a novel distributed controller design methodology for multi-agent systems over a class of directed graphs by extending the gradient-flow method.
First, we set the applicable class of directed graphs and provided a novel gradient-based distributed controller.
Next, we presented the convergence analysis and numerical results, which exhibited better performance than the conventional gradient-flow method.
Furthermore, we applied this methodology to the dynamic matching problem of two agent groups with different sensing ranges, modeled by a directed proximity graph, and consequently derived sufficient conditions for successful matching.
Finally, numerical experiments verified the effectiveness of the proposed controller in the case of the directed proximity graph.
% A future direction would be to construct a more general distributed controller design methodology applicable to a broader class of directed graphs.
\appendix
\subsection{Proof of Lemma \ref{lem:eql}}\label{subsec:lem_eql}
A preliminary on differential inclusions \cite{filippov1988differential,sakurama2018multiagent,shevitz1994lyapunov,bacciotti1999stability,smirnov2022introduction} is provided below.
% \begin{defi}

    Consider
    % $\dot{X}(t)=F(X(t))$,
    \begin{equation}\label{eq:DE}
        \dot{X}(t)=F(X(t)),
    \end{equation}
    where $X(t)\in\BRdn$
    and $F:\BRdn\to\BRdn$ is essentially locally bounded, but not necessarily continuous.
    $X(t)\in\BRdn$ is called a \textit{Filippov solution} of \eqref{eq:DE} if $X(t)$ is absolutely continuous and satisfies the differential inclusion:
\begin{equation}
% \label{inc}
    \label{XtinKf}
    \dot{X}(t) \in \mathcal{K}[F](X(t)).
\end{equation}
Here, $\mathcal{K}[F]$ represents the set-valued mapping defined by
% \begin{align*}
   $ \mathcal{K}[F](X) := \mathrm{cl}\,\mathrm{co}\bigl\{
   % &
   \lim_{k\to \infty}F(X_k)\in \mathbb{R}^{d\times n}:
    % \\
 % &
 \{X_k\}\subset \mathbb{R}^{d\times n}\setminus \mathcal{U}
\text{ s.t. }\lim_{k\to \infty}X_k=X \bigr\}$
% \end{align*}
with a set $\mathcal{U}$ of measure zero, where $\mathrm{cl}\,\mathrm{co}$ denotes the closure of the convex hull of a set. 
% \end{defi}
Next, 
% for a differentiable function $V$, the set-valued derivative is defined as follows.
% \begin{defi}\cite{bollobas1998modern,shevitz1994lyapunov}
    let a function $V:\BRdn\to\mathbb{R}$ be differentiable. Then, the \textit{set-valued derivative} $\dot{\tilde{V}}(X)$ of $V$ with respect to \eqref{XtinKf}
    is defined as
    % \begin{align*}
       $\textstyle \small \dot{\tilde{V}}(X) := \biggl\{v=\sum_{i=1}^{n}(\nabla_i V(X))^\top f: f\in \mathcal{K}[F](X)\biggr\} \subset \mathbb{R}.$
    % \end{align*}
% \end{defi}
Besides, 
% for differential inclusions, a \textit{weakly invariant set} is defined as follows.
% \begin{defi}\cite{bacciotti1999stability}
    a set $\Omega\subset\BRdn$ is said to be a \textit{weakly invariant} set for \eqref{eq:DE} if there exists a maximal solution of \eqref{eq:DE} lying in $\Omega$ through each $X(0)\in\Omega$.
% \end{defi}
% \subsubsection{Proof}

Now, we prove Lemma \ref{lem:eql} as follows. 
Consider $\kappa_i>0$ for all $i\in\CV_\mathrm{t}$.
Without loss of generality, let $\CV_\mathrm{t}=\{1,2,\ldots,p\}$ for some positive integer $p\,(<n)$.
From \eqref{gbar_explicit}, $\bar{g}_i(X)$ can be discontinuous.
We consider the Filippov solution of the differential inclusion in \eqref{XtinKf}.
% \begin{equation}
%     \label{XtinKf}
% \dot{X}(t) \in \mathcal{K}[F](X(t)),
% \end{equation}
Here, by \eqref{dynamics}, \eqref{controller}, and \eqref{ctr:digraph}, $F(X)$ is reduced to
\begin{align}
\label{eq:FX}
\small
    &F(X) 
    % \nonumber\\
    =\biggl[-\kappa_1 \nabla_1 V_\mathrm{ud} (X) -\bar{g}_1(X),
    \ldots,
    % \nonumber\\
    -\kappa_{p} \nabla_{p} V_\mathrm{ud} (X) 
    % -\bar{g}_{p}(X), 
    \nonumber \\
& -\bar{g}_{p}(X), -\mu_{p+1} \nabla_{p+1} V_\mathrm{ud} (X), \ldots, -\mu_{n} \nabla_{n} V_\mathrm{ud} (X) \biggr].
\end{align}
\textcolor{black}{The existence of the solution is guaranteed for the differential inclusion \eqref{XtinKf} with \eqref{eq:FX} by Theorem 1 in Section 7 in Chapter 2 of \cite{filippov1988differential}.
Note that the set-valued map $\mathcal{K}[F](X)$ with $F(X)$ is closed, convex, and upper semi-continuous from the assumption of the continuous differentiability of $V_\mathrm{ud}(X)$ and $\bar{V}_\mathrm{ud}$ and Proposition 2.2 in \cite{smirnov2022introduction}.} 

By \eqref{XtinKf} with \eqref{eq:FX}, the set-valued derivative $\dot{\tilde{V}}(X)$ is
% \begin{align*}
% \small
    % \label{setderiVB}
$\dot{\tilde{V}}_\mathrm{ud}(X) =\biggl\{v=\sum_{i\in \CN}(\nabla_i V_\mathrm{ud}(X) )^\top h_i
% \in \mathbb{R} 
% \nonumber \\
:[h_j]_{j\in\CN}\in \mathcal{K}[F](X)\biggr\}$.
% \end{align*}
From \eqref{gibar:c}, we obtain $(\nabla_i V_\mathrm{ud} (X))^\top\bar{g}_i(X)\geq 0,\forall i\in \CV_\mathrm{t}$; thus $ v\leq 0$ for all $\forall v\in \dot{\tilde{V}}_\mathrm{ud}(X)\subset \mathbb{R}$.
In addition, $0\in \dot{\tilde{V}}_\mathrm{ud}(X) \Rightarrow \nabla_i V_\mathrm{ud} (X)=0 ,\forall i \in \CN$ because we have
\begin{align}
\small
% &\sum_{i\in \CN}(\nabla_i V_\mathrm{ud}(X) )^\top \dot{x}_i  \nonumber \\
% =& -\sum_{i\in \CN\setminus\CV_\mathrm{t}} \mu_{i}\|\nabla_i V_\mathrm{ud}(X) \|^2 - \sum_{i\in \CV_\mathrm{t}}\kappa_i \|\nabla_i V_\mathrm{ud}(X) \|^2  \nonumber\\
% \label{deri}
% &- \sum_{i\in \CV_\mathrm{t}} \left(\nabla_i V_\mathrm{ud}(X) \right)^\top \bar{g}_i(X)
&\sum_{i\in \CN}(\nabla_i V_\mathrm{ud}(X) )^\top \dot{x}_i  
= -\sum_{i\in \CN\setminus\CV_\mathrm{t}} \mu_{i}\|\nabla_i V_\mathrm{ud}(X) \|^2 \nonumber \\
&- \sum_{i\in \CV_\mathrm{t}}\kappa_i \|\nabla_i V_\mathrm{ud}(X) \|^2 
- \sum_{i\in \CV_\mathrm{t}} \left(\nabla_i V_\mathrm{ud}(X) \right)^\top \bar{g}_i(X)
\label{deri}
\end{align}
and $(\nabla_i V_\mathrm{ud} (X))^\top\bar{g}_i(X)\geq 0$ from \eqref{gibar:c}.
Therefore, from Theorem 3 in \cite{bacciotti1999stability}, the Filippov solution $X(t)$ of \eqref{XtinKf} with \eqref{eq:FX} converges to the largest weakly invariant set $\mathcal{M}$ in $\nabla V_\mathrm{ud}^{-1}(0)\cap L_{V_\mathrm{ud}}( V_\mathrm{ud}(X(0) )=\{X:0\in \dot{\tilde{V}}_\mathrm{ud}(X)\}\cap L_{V_\mathrm{ud}}( V_\mathrm{ud}(X(0) )$.
% from the assumption of boundedness of $L_{V_\mathrm{ud}}( V_\mathrm{ud}(X(0) )$.

Next, we consider the solution lying in $\nabla V_\mathrm{ud}^{-1}(0) \cap L_{V_\mathrm{ud}}( V_\mathrm{ud}(X(0) )$.
Then, $\nabla_i V_\mathrm{ud} (X)=0,\forall i\in \CN$ over $\nabla V_\mathrm{ud}^{-1}(0)$ holds.
Hence, by \eqref{dynamics}, \eqref{controller}, and \eqref{ctr:digraph}, 
we obtain 
% \begin{equation*}
$\dot{x}_i(t)
    % \dot{x}_i(t) 
    = -\frac{\lambda_i}{2} \nabla_i \bar{V}_\mathrm{ud} (X(t)),\; i\in \CV_\mathrm{t}$,
 %    \begin{cases}
 % -\frac{\lambda_i}{2} \nabla_i \bar{V}_\mathrm{ud} (X(t)),& i\in \CV_\mathrm{t}\\
 % 0,&i\in \CN\setminus\CV_\mathrm{t}
 % \end{cases}
% \end{equation*}
and $\dot{x}_i(t) =0,\,i\in \CN\setminus\CV_\mathrm{t}$
for $X(t)=[x_1(t),\ldots,x_n(t)]$ over $\nabla V_\mathrm{ud}^{-1}(0)$.
Then, we have $\dot{\bar{V}}_\mathrm{ud}(X(t))= - \sum_{i\in \CV_\mathrm{t}}\frac{\lambda_i}{2} \|\nabla_i \bar{V}_\mathrm{ud} (X(t))\|^2\leq 0 $
% \begin{align*}
% % \label{VAdot}
%  \dot{\bar{V}}_\mathrm{ud}(X(t)) 
%  % &= \sum_{i\in \CV_\mathrm{t}} (\nabla_i \bar{V}_\mathrm{ud} (X(t)))^\top \dot{x}_i \nonumber \\
%  % &
%  = - \sum_{i\in \CV_\mathrm{t}}\frac{\lambda_i}{2} \|\nabla_i \bar{V}_\mathrm{ud} (X(t))\|^2\leq 0
% \end{align*}
over $\nabla V_\mathrm{ud}^{-1}(0)\cap L_{V_\mathrm{ud}}( V_\mathrm{ud}(X(0) )$.
Hence, from Lasalle's invariant theorem \cite{khalil2014nonlinear}, the solution $X(t)$ lying $\nabla V_\mathrm{ud}^{-1}(0)\cap L_{V_\mathrm{ud}}( V_\mathrm{ud}(X(0) )$ converges to $\{X\in \mathbb{R}^{d\times n}:\dot{\bar{V}}_\mathrm{ud}(X)=0\}$.
Thus, the largest weakly invariant set in $\nabla V_\mathrm{ud}^{-1}(0)\cap L_{V_\mathrm{ud}}( V_\mathrm{ud}(X(0) )$ is reduced to $\mathcal{M}=\{X\in \mathbb{R}^{d\times n}:\dot{\bar{V}}_\mathrm{ud}(X)=0\}\cap \nabla V_\mathrm{ud}^{-1}(0)\cap L_{V_\mathrm{ud}}( V_\mathrm{ud}(X(0) )$.
Now, from $\lambda_i>0$, we obtain $\dot{\bar{V}}_\mathrm{ud}(X)=0\Leftrightarrow \nabla_i \bar{V}_\mathrm{ud} (X)=0,\forall i\in \CV_\mathrm{t}$ when $X\in\nabla V_\mathrm{ud}^{-1}(0)$.
Consequently, $X(t)\to\mathcal{M}=\nabla V_\mathrm{ud}^{-1}(0)\cap\nabla_{\CV_\mathrm{t}}\bar{V}_\mathrm{ud}^{-1}(0)\cap L_{V_\mathrm{ud}}( V_\mathrm{ud}(X(0) )$ holds.
Therefore, $\Omega=\nabla V_\mathrm{ud}^{-1}(0)\cap\nabla_{\CV_\mathrm{t}}\bar{V}_\mathrm{ud}^{-1}(0)$ is globally attractive.

\subsection{Proof of Corollary \ref{coro:eql}}\label{subsec:coro}
By Lemma \ref{lem:eql}, it is sufficient to prove the case of $\kappa_i=0$.
% Now, for the Filippov solution \eqref{XtinKf} with \eqref{eq:FX}, we prove Corollary \ref{coro:eql}. 

Assume that $\kappa_i=0$, and that $\nabla_i V_\mathrm{ud} (X)=0,\forall i\in \CN\setminus\CV_\mathrm{t}
\Leftrightarrow\nabla_i V_\mathrm{ud} (X)=0, \forall i\in \CV_\mathrm{t}$ holds.
Then, the inequality in \eqref{deri} is reduced to $\sum_{i\in \CN}(\nabla_i V_\mathrm{ud}(X) )^\top \dot{x}_i 
= -\sum_{i\in \CN} \mu_{i}\|\nabla_i V_\mathrm{ud}(X) \|^2 
- \sum_{i\in \CV_\mathrm{t}} (\nabla_i V_\mathrm{ud}(X) )^\top \bar{g}_i(X)$.
% \begin{align*}
% &\sum_{i\in \CN}(\nabla_i V_\mathrm{ud}(X) )^\top \dot{x}_i \\
% =& -\sum_{i\in \CN} \mu_{i}\|\nabla_i V_\mathrm{ud}(X) \|^2 
% - \sum_{i\in \CV_\mathrm{t}} (\nabla_i V_\mathrm{ud}(X) )^\top \bar{g}_i(X).
% \end{align*}
Thus, when $0\in \dot{\tilde{V}}_\mathrm{ud}(X)$, we have $\nabla_i V_\mathrm{ud} (X)=0$ for all $i\in\CN$.
% by the additional assumption.
The statement follows from the same argument as Lemma \ref{lem:eql}.
% Following the same discussion as the proof of Lemma \ref{lem:eql} in Appendix \ref{subsec:lem_eql}, we can prove the global attractiveness of $\Omega$.
\subsection{Proof of Theorem \ref{thm:distributed}}\label{subsec:thm_distributed}

Here, we prove that the generalized controller in Remark \ref{rem:gen}:
\begin{equation}
% \small
\label{ctr:digraph_gen}
 f_i(X) = \begin{cases}
  -\bar{g}_i(X) -\kappa_{i}\displaystyle \nabla_i V_\mathrm{ud}(X) , & i \in \CN_\mathrm{A} \\
  -\mu_i\displaystyle \nabla_i V_\mathrm{ud}(X) , & i\in \CN_\mathrm{B}
 \end{cases}
\end{equation}
is distributed with $\CNRA$ in \eqref{def:na}, $\CNRB$ in \eqref{def:nb}, $\bar{G}_\mathrm{ud}=(\CN,\CE_\RA)$ in \eqref{condition:ea}, and $G_\mathrm{ud}=(\CN,\CE_\mathrm{B})$ in \eqref{condition:eb}.

First, by \eqref{condition:eb}, we obtain $\CN_i(G_\mathrm{ud})\subset\CN_i(G)$.
Then, from Theorem 1 in \cite{sakurama2014distributed}, $\nabla_i V_\mathrm{ud} (X)$ is distributed for any $i \in \CN$.
Thus, for any $i \in \CNRB$, the controller $f_i(X)$ in \eqref{ctr:digraph_gen} is distributed. 
Next, we consider $i \in \CNRA$.
In the following, we prove that $\nabla_i \bar{V}_\mathrm{ud} (X)$ is distributed for any $ i \in \CNRA$ because
% if $\nabla_i \bar{V}_\mathrm{ud} (X)$ is distributed for any $i \in \CNRA$, then 
in this case,
the function $\bar{g}_i(X)$ in \eqref{gbar_explicit} is also distributed for any $ i\in\CNRA$ due to the distributedness of $\nabla_i V_\mathrm{ud} (X)$ for any $i \in \CNRA$.
% Then, the controller in \eqref{ctr:digraph_gen} is distributed for any $i \in \CN$.
% Therefore, we prove that $\nabla_i V_\RA (X)$ is distributed for any $ i \in \CNRA$ in the following way.
From the gradient-distributedness of $\bar{V}_\mathrm{ud}$ with respect to $\bar{G}_\mathrm{ud}$ and Theorem 1 in \cite{sakurama2014distributed}, for any $ i \in \CNRA$, there exists a function $h_i:\mathbb{R}^{d\times (|\CN_i(G_\mathrm{A})|+1)}\rightarrow \mathbb{R}^d$ satisfying the distributedness with respect to $\bar{G}_\mathrm{ud}$, i.e., $\nabla_i \bar{V}_\mathrm{ud} (X)=h_i(x_i,[x_j]_{j\in \CN_i(\bar{G}_\mathrm{ud})})$.
% \begin{equation}
%     \label{hhat}
% \nabla_i \bar{V}_\mathrm{ud} (X)=h_i(x_i,[x_j]_{j\in \CN_i(\bar{G}_\mathrm{ud})}).
% \end{equation}
Now, let $\check{G}$ and $\widehat{G}$ be $\check{G}=(\CN,\bar{\CE}_\mathrm{ud})$ and $\widehat{G}=(\CN,\bar{\CE}_\mathrm{ud}\cup\{(i,j):\forall i,j\in\CN_\mathrm{B},i\neq j\})$, respectively. Then, $\CN_i(\widehat{G})=\CN_i(\check{G}),\,\forall i\in\CNRA$ holds.
Combining this and \eqref{condition:ea} gives $\CN_i(\bar{G}_\mathrm{ud})=\CN_i(\widehat{G})=\CN_i(\check{G}),\, \forall i \in \CNRA.$
% \begin{equation}
% \label{eq:ni1}
% \CN_i(\bar{G}_\mathrm{ud})=\CN_i(\widehat{G})=\CN_i(\check{G})\quad \forall i \in \CNRA.
% \end{equation}
Furthermore, from \eqref{def:v1}, \eqref{def:na}, and \eqref{def:ebar}, no node in $\CNRA$ is a head of unidirectional edges in $G$, and hence we have $\CN_i(G)=\CN_i(\check{G}),\, \forall i \in \CNRA$.
% \begin{equation}
% \label{eq:ni2}
% \CN_i(G)=\CN_i(\check{G})\quad \forall i \in \CNRA.
% \end{equation}
Then, this equation and 
% \eqref{eq:ni1} 
$\CN_i(\bar{G}_\mathrm{ud})=\CN_i(\widehat{G})=\CN_i(\check{G})$
yield $\CN_i(G)=\CN_i(\bar{G}_\mathrm{ud})$ for any $ i \in \CNRA$.
Consequently, from the gradient-distributedness of $\nabla_i \bar{V}_\mathrm{ud} (X)$ with respect to $\bar{G}_\mathrm{ud}$, we obtain
% \begin{equation*}
$\nabla_i \bar{V}_\mathrm{ud} (X)=h_i(x_i,[x_j]_{j\in \CN_i(G)})$ for any $i \in \CNRA$.
% \end{equation*}
Therefore, $\nabla_i \bar{V}_\mathrm{ud} (X)$ is distributed for any $ i \in \CNRA$.
\subsection{Proof of Theorem \ref{thm:ctr_matching_attractive}}\label{subsec:proof_thm_attraction}

Here, we show the attractiveness of $\CS= \CT\setminus(\CU_1\cup\CU_2)$ with $\CT$ in \eqref{def:target_set}, $\CU_1$ in \eqref{def:overlapping}, and $\CU_2$ in \eqref{def:margin_delta_K}.
% Note that $\CS$ satisfies $\bar{\CS}=\CT$.

First, let $\CA = \bigcup_{Y\in\CS} \CA_Y$ with 
\begin{equation}
    \label{def:ca_y}
\CA_Y= \{X\in\mathbb{R}^{d\times n}:\|x_i-y_i\| < \varepsilon(Y)\;\forall i\in\CN\}.
\end{equation}
Here, we define the nonnegative function $\varepsilon(Y)$ as follows:
\begin{equation}\label{epY}
%\small
     \varepsilon(Y)=\min\biggl\{\varepsilon_0(Y),\varepsilon_\mathrm{A}(Y),\varepsilon_\mathrm{B}(Y),\delta_{\mathrm{B}}/2\biggr\},
\end{equation}
where 
\begin{equation}
     \label{ep0Y}
 %\small
 \varepsilon_0(Y)= \frac{1}{6} \min_{\substack{i\in \mathcal{N}_{\mathrm{mn}},\\j \in \mathcal{N}_{\mathrm{mj}}, j\neq \alpha_Y(i)}}\|y_i-y_j\|\quad
\end{equation}
\begin{equation}
    \label{epKY}
\varepsilon_\mathrm{K}(Y)= \frac{1}{2}\min_{\substack{i\in \mathcal{N}_{\mathrm{mn}},\\j \in \mathcal{N}_{\mathrm{mj}},  j\neq \alpha_Y(i)}}|\delta_{\mathrm{K}}-\|y_i-y_j\||,\, \mathrm{K}=\mathrm{A},\mathrm{B}
\end{equation}
% \begin{align}
%  \label{ep0Y}
%  %\small
%  \varepsilon_0(Y)&= \frac{1}{6} \min_{\substack{i\in \mathcal{N}_{\mathrm{mn}},\\j \in \mathcal{N}_{\mathrm{mj}}, j\neq \alpha_Y(i)}}\|y_i-y_j\|,\\
%  \label{epKY}
% \varepsilon_\mathrm{K}(Y)&= \frac{1}{2}\min_{\substack{i\in \mathcal{N}_{\mathrm{mn}},\\j \in \mathcal{N}_{\mathrm{mj}},  j\neq \alpha_Y(i)}}|\delta_{\mathrm{K}}-\|y_i-y_j\||,\, \mathrm{K}=\mathrm{A},\mathrm{B}
% \end{align}
with  $Y=[y_1,\ldots,y_n]\in\BRdn$ and
\begin{align}
\label{alpha_Y}
\alpha_Y \in \argmin_{\alpha\in\Pi(\mathcal{N}_{\mathrm{mn}},\mathcal{N}_{\mathrm{mj}})} \|[y_j]_{j\in\mathcal{N}_{\mathrm{mn}}} - [y_{\alpha(j)}]_{j\in\mathcal{N}_{\mathrm{mn}}}\| .
\end{align}
% If $\epsilon(Y)>0$ for any $Y\in\CS$, then $\CA$ is open and contains $CT$.

By the following lemma, $\varepsilon(Y)>0$ holds for any $Y\in\CS$.
Then, $\CA=\bigcup_{Y\in\CS}$ with \eqref{def:ca_y} is an open set containing $\CS$.
\begin{lemma}\label{lem:ep}
The nonnegative scalar valued function $\varepsilon(Y)\geq 0$ in \eqref{epY}, \eqref{ep0Y} and \eqref{epKY} is positive for $Y=[y_1,\ldots,y_n]\in\CS$.
% if and only if $\|y_i-y_j\|\neq \delta_\mathrm{A},\delta_\mathrm{B}$ holds for any $i\in\Nmn,\,j\in\Nmj$ such that $\alpha_Y(i)\neq j$.
\end{lemma}
\begin{proof}
% Suppose that $\|y_i-y_j\|\neq \delta_\mathrm{A},\delta_\mathrm{B}$ holds for any $i\in\Nmn,\,j\in\Nmj$ satisfying $\alpha_Y(i)\neq j$.
Let $Y\in\CS=\CT\setminus(\CU_1\cup\CU_2$
From \eqref{def:target_set} and \eqref{ep0Y}, we have $\varepsilon_0(Y)>0$ for any $Y\in\CS$. In addition, from \eqref{epKY}, we obtain $\varepsilon_\RK(Y)>0$.
Therefore, $\varepsilon(Y)>0$ follows from \eqref{epY}.
% Conversely, if $\varepsilon(Y)>0$, we have $\varepsilon_\RK(Y)>0$.
% Then, from \eqref{epKY}, $Y\in\CT$ satisfies $\|y_i-y_j\|\neq \delta_\mathrm{A},\delta_\mathrm{B}$  for any $i\in\Nmn,\,j\in\Nmj$ satisfying $\alpha_Y(i)\neq j$.
\end{proof}

Then, it is sufficient that we prove that $\CA$ becomes a region of attraction of $\CS$.
Now, we have the following two lemmas.
\begin{lemma}\label{lem:attraction1}
% Suppose that $\varepsilon:\CT\to \mathbb{R}_+\cup\{0\}$ is given as Lemma \ref{lem:ep}.
If $X\in\CA_Y$ holds for $\CA_Y$ in \eqref{def:ca_y} with some $Y\in\CS$ satisfying $\varepsilon
(Y)>0$, then $\Cmn\subset\Nmn$ and $\Cmj\subset\Nmj$ hold for any $ \mathcal{C} \in \mathrm{clq}(G_{\delta_\mathrm{K}}(X)),\mathrm{K}=\mathrm{A},\mathrm{B}$ such that $\mathcal{C}_{\mathrm{mn}}\neq \emptyset$. Furthermore, $\alpha_Y\in\Pi(\Nmn,\Nmj)$ in \eqref{alpha_Y} fulfills $
% \begin{equation}
\alpha_Y(j)\in \Cmj$ for any $j\in\Cmn.$
% \end{equation}
\end{lemma}

\begin{proof}
Let $\RK$ be $\RA$ or $\RB$.
For $X\in\CA_Y$ and any $i,j\in \CN$,
\begin{equation}
\label{xyinq}
 % \|y_i-y_j\|-2\varepsilon(Y)\leq\|x_i-x_j\|\leq\|y_i-y_j\|+2\varepsilon(Y)
\|y_i-y_j\|-2\varepsilon(Y)<\|x_i-x_j\|<\|y_i-y_j\|+2\varepsilon(Y)
\end{equation}
holds. Then, for $i\in\Nmn$ and $j\in \Nmj$ satisfying $j\neq \alpha_Y(i)$, from \eqref{epKY},
$(\|y_i-y_j\|-2\varepsilon(Y))-\|y_i-y_j\|
< (\|y_i-y_j\|+2\varepsilon(Y)) - \|y_i-y_j\| 
\leq|\delta_{\mathrm{K}}-\|y_i-y_j\|| \neq 0$.
Thus, we obtain \begin{align}
%\small
    \label{y_delta_in}
\|y_i-y_j\|<\delta_{\mathrm{K}} &\Rightarrow \|y_i-y_j\|+2\varepsilon(Y) < \delta_{\mathrm{K}}  \normalsize \\
\label{y_delta_out}
%\small
\|y_i-y_j\|>\delta_{\mathrm{K}} &\Rightarrow \|y_i-y_j\|-2\varepsilon(Y) > \delta_{\mathrm{K}} \normalsize
\end{align}
are satisfied.
First, we prove the former part of Lemma \ref{lem:attraction1}. Suppose that, for some $ \mathcal{C} \in \mathrm{clq}(G_{\delta_\mathrm{K}}(X))$ satisfying $\Cmn\neq\emptyset$, $\Cmn\subset\Nmj$ and $\Cmj\subset\Nmn$ hold.
Then, we have $\alpha_Y(l)\notin\Cmn$ for some $l\in\Cmj$.
Thus, we obtain $\|x_{\alpha_Y(l)} - x_m\|>\delta_\RK$ for some $m\in\CC$.
In addition, from $l,m\in\CC$ and \eqref{xyinq}, we have $\|y_l-y_m\|-2\varepsilon(Y)=\|y_{\alpha_Y(l)}-y_m\|-2\varepsilon(Y)
% \leq\|x_m-x_l\|\leq\delta_{\mathrm{K}}$.
<\|x_m-x_l\|\leq\delta_{\mathrm{K}}$.
Then, from \eqref{y_delta_out}, $\|y_l-y_m\|=\|y_{\alpha_Y(l)}-y_m\|\leq \delta_{\mathrm{K}}$ follows.
Therefore, using \eqref{epKY} and \eqref{y_delta_in}, we obtain $\delta_{\mathrm{K}}< \|x_{\alpha_Y(l)} - x_m\| < \|y_{\alpha_Y(l)} - y_m\| +2\varepsilon(Y)=\|y_l - y_m\| +2\varepsilon(Y)<\delta_{\mathrm{K}}$.
This inequality has a contradiction.

Next, we prove the latter part.
% of the lemma.
Assume that $\alpha_Y(\Cmn)\nsubseteq \Cmj$ for some $\CC$.
Then, for some $l\in\Cmn$, we have $\alpha_Y(l)\notin \Cmj$ and there exists $m\in\CC$ such that $\|x_{\alpha_Y(l)}-x_m\|>\delta_{\mathrm{K}}$.
Besides, $\|x_l-x_m\|\leq\delta_{\mathrm{K}}$ for $l,m\in \CC$. Hence, we have $\|y_l-y_m\|=\|y_{\alpha_Y(l)}-y_m\|<\delta_{\mathrm{K}}$ from \eqref{epKY}, \eqref{xyinq}, and the contraposition of \eqref{y_delta_out}.
Then, $\|x_{\alpha_Y(l)}-x_m\|<\delta_{\mathrm{K}}$ follows from \eqref{xyinq} and \eqref{y_delta_in}, which contradicts $\|x_{\alpha_Y(l)}-x_m\|>\delta_{\mathrm{K}}$.
% The above discussion does not rely on the value of $\RK$.
\end{proof}

\begin{lemma}\label{lem:attraction2}
Suppose that $X \in \CA_Y$ holds for some $Y\in\CS$.
% with $\varepsilon:\CT\to \mathbb{R}_+\cup\{0\}$ in Lemma \ref{lem:ep}.
Then, for $\mathcal{C} \in \mathrm{clq}(G_{\delta_\mathrm{K}}(X)),\mathrm{K}=\mathrm{A},\mathrm{B}$ satisfying $\Cmn\neq\emptyset$ and $\bar{\alpha}_{\mathcal{C}}$ in \eqref{def:alphaC}, the equation
\begin{align}
\label{alphaYc}
\alpha_Y|_{\Cmn} = \bar{\alpha}_{\mathcal{C}}
\end{align}
holds, where $\alpha_Y|_{\Cmn}$ represents a restricted mapping of $\alpha_Y$ in \eqref{alpha_Y} for $\Cmn$, i.e., $\alpha_Y(j)=\alpha_Y|_{\Cmn}(j)$ for all $j \in \Cmn$.
\end{lemma}
\begin{proof}
Let $\RK$ be an element of $\{\RA,\RB\}$.
Suppose that $\bar{\alpha}_\CC\neq \alpha_Y|_{\Cmn}$ holds for some $ \CC \in \mathrm{clq}(G_{\delta_\mathrm{K}}(X))$ satisfying $\Cmn\neq\emptyset$.
Then, we obtain $\|[x_j]_{j\in \Cmn}-[x_{\bar{\alpha}_\CC(j)}]_{j\in \Cmn}\|<\|[x_j]_{j\in \Cmn}-[x_{\alpha_Y(j)}]_{j\in \Cmn}\|$ since $\Cmn\cup\alpha_Y|_{\Cmn}(\Cmn)\subset \CC$ holds from Lemma \ref{lem:attraction1}.
Thus, for some $j\in\Cmn$, $\|x_j-x_{\bar{\alpha}_\CC(j)}\|<\|x_j-x_{\alpha_Y(j)}\|<2\varepsilon(Y)$ is obtained.
Meanwhile, from \eqref{epY} and \eqref{ep0Y}, we obtain
$2\varepsilon(Y) > \|x_j-x_{\bar{\alpha}_\CC(j)}\| 
\geq \|x_{\bar{\alpha}_\CC(j)}-x_{\alpha_Y(j)}\| - \|x_j-x_{\alpha_Y(j)}\|
> \|y_{\bar{\alpha}_\CC(j)}-y_{\alpha_Y(j)}\| -2\varepsilon(Y) -2\varepsilon(Y) 
\geq \min_{l\in\Cmj,l\neq\alpha_Y(j)}\|y_j-y_l\|-4\varepsilon(Y)$.
% \begin{align}
% %\small
%     2\varepsilon(Y) >& \|x_j-x_{\bar{\alpha}_\CC(j)}\| \nonumber \\
% % =&\|x_j-x_{\alpha_Y(j)} +x_{\alpha_Y(j)}-x_{\bar{\alpha}_\CC(j)}\| \nonumber\\
% \geq& \|x_{\bar{\alpha}_\CC(j)}-x_{\alpha_Y(j)}\| - \|x_j-x_{\alpha_Y(j)}\| \nonumber \\
% >& \|y_{\bar{\alpha}_\CC(j)}-y_{\alpha_Y(j)}\| -2\varepsilon(Y) -2\varepsilon(Y) \nonumber\\
% \geq& \min_{l\in\Cmj,l\neq\alpha_Y(j)}\|y_j-y_l\|-4\varepsilon(Y) \nonumber
% \end{align}
Hence $\min_{l\in\Cmj,l\neq\alpha_Y(j)}\|y_j-y_l\| < 6\varepsilon(Y) \leq\min_{l\in\Cmj,l\neq\alpha_Y(j)}\|y_j-y_l\|$ follows from \eqref{ep0Y}, which yields a contradiction.
% The discussion does not rely on the value of $\RK$.
Thus, \eqref{alphaYc} holds for $\mathcal{C} \in \mathrm{clq}(G_{\delta_\mathrm{K}}(X))$.
\end{proof}
% \begin{equation*}
%      \bar{\alpha}_{\mathcal{C}}= \argmin_{\alpha_{\mathcal{C}} \in \Pi(\mathcal{C}_{\mathrm{mn}},\mathcal{C}_{\mathrm{mj}})} 
%  \|[x_j]_{j\in \mathcal{C}_{\mathrm{mn}}} 
%  - 
%  [x_{\alpha_{\mathcal{C}}(j)}]_{j\in \mathcal{C}_{\mathrm{mn}}}\| .
% \end{equation*}

Using Lemmas \ref{lem:ep}, \ref{lem:attraction1}, and \ref{lem:attraction2}, we prove Theorem \ref{thm:ctr_matching_attractive}. Assume $X(0)\in\CA$ is satisfied.
Then, there exists $Y\in\CS$ such that $X(0)\in\CA_Y$.
% First, if $\varepsilon(Y)=0$,
% then, $\CA_Y\subset\CT$ holds, which yields $X(0)\in\CT$.
% Hence, $X(t)\to\CT$ is achieved because $V_\RA(X)=V_\RB(X)=0$ holds for $X\in\CT$, which yields $f_i(X)=0,\,i\in\CN$.
% Next, we consider the case of $\varepsilon(Y)>0$.
From Lemmas \ref{lem:ep} and \ref{lem:attraction2}, if $X(0) \in \CA_Y$, then a restricted mapping $\alpha_Y|_{\Cmn}$ can always be defined for any $\mathcal{C} \in \mathrm{clq}(G_{\delta_\mathrm{K}}(X)),\mathrm{K}=\mathrm{A},\mathrm{B}$ satisfying $\mathcal{C}_{\mathrm{mn}}\neq \emptyset$.
Thus, from \eqref{alphaYc}, the dynamics in \eqref{dynamics} with \eqref{controller} and the designed controller is reduced to
% \begin{equation}
% % \footnotesize
%     \label{dynamics_in_AY}
%  \dot{x}_i(t) = 
%  \begin{cases}
%      -\lambda_i(X(t)) (x_i(t)
%       -x_{\alpha_Y(i)}(t)), &\!\!\! i\in \Nmn\\
%      -\lambda_i(X(t)) (x_i(t)
%       -x_{\alpha_Y^{-1}(i)}(t)),&\!\!\! i\in\Nmj\cap\alpha_Y(\mathcal{N}_{\mathrm{mn}})\\
%   0,&\!\!\! i\in\Nmj\setminus\alpha_Y(\mathcal{N}_{\mathrm{mn}}),
%  \end{cases}\normalsize
% \end{equation}
\begin{equation}
% \small
    \label{dynamics_in_AY}
 \dot{x}_i(t) = 
 \begin{cases}
 &-\zeta_i(X(t)) (x_i(t)
      -x_{\alpha_Y(i)}(t)), \\
 &\quad\quad\quad\quad\quad\quad\quad i\in \Nmn\\
     &-\zeta_i(X(t)) (x_i(t)
      -x_{\alpha_Y^{-1}(i)}(t)),\\
 &\quad\quad\quad\quad\quad\quad\quad i\in \Nmj\cap\alpha_Y(\mathcal{N}_{\mathrm{mn}})\\
  &0,\quad\quad\quad\quad\quad\quad i\in\Nmj\setminus\alpha_Y(\mathcal{N}_{\mathrm{mn}}),
 \end{cases}\normalsize
\end{equation}
% \begin{equation}
% \small
%     \label{dynamics_in_AY}
%  \dot{x}_i(t) = 
%  \begin{cases}
%  \begin{split}
%      - \eta_i(X(t&)) (x_i(t) \\
%       -&x_{\alpha_Y(i)}(t)),
%  \end{split}
%  &i\in \Nmn\\
%  \begin{split}
%      - \eta_i(X(t&)) (x_i(t) \\
%       -&x_{\alpha_Y^{-1}(i)}(t)),
%  \end{split}
%  &i\in \Nmj\cap\alpha_Y(\mathcal{N}_{\mathrm{mn}})\\
%   0, 
%   &i\in\Nmj\setminus\alpha_Y(\mathcal{N}_{\mathrm{mn}}).
%  \end{cases}\normalsize
% \end{equation}
where $\zeta_i(X)$ is given as follows:
\begin{equation}
\label{lambda}
% \small
\zeta_i(X) = \begin{dcases}
% \begin{split}
%     &\frac{1}{2} \bigl\{\mu_{\mathrm{A}i} |\mathrm{clq}_i(G_{\delta_\mathrm{A}}(X))|\\
%     &\:+ (2\kappa_i+ \mu_{\mathrm{B}i} )|\mathrm{clq}_i(G_{\delta_\mathrm{B}}(X))| \bigr\} ,
% \end{split}
\frac{1}{2} \bigl(\lambda_i |\mathrm{clq}_i(G_{\delta_\mathrm{A}}(X))|& \\
    \:+ (2\kappa_i+ \eta_i )|\mathrm{clq}_i(G_{\delta_\mathrm{B}}(X))| \bigr) , &\!\! i\in\mathcal{N}_\mathrm{A}\\
 \mu_i |\mathrm{clq}_i(G_{\delta_\mathrm{B}}(X))| , &\!\! i\in\mathcal{N}_\mathrm{B}.
 \end{dcases}\normalsize
\end{equation}
Note that $\zeta_i(X)$ is a positive number for $i\in \CN \setminus\alpha_Y(\mathcal{N}_{\mathrm{mn}})$ because $\|x_i-x_{\alpha_Y(i)}\|< 2\varepsilon(Y)\leq \delta_\RB$ holds on $\CA_Y$.
From \eqref{dynamics_in_AY}, for $i\in\Nmn$ and $\alpha_Y(i)\in\Nmj$, we obtain
\begin{align}
% \small
\label{dynamics:displacement}
\dot{x}_i(t)-\dot{x}_{\alpha_Y(i)}&(t) =-(\zeta_i(X(t)) \nonumber\\
&+\zeta_{\alpha_Y(i)}(X(t)))(x_i(t)-x_{\alpha_Y(i)}(t)).
\normalsize
\end{align}

Now, we prove $X(t)\to\CT$ as follows.
From \eqref{dynamics:displacement}, agents $i$ and $\alpha_Y(i)$ approach each other along with the line segment connecting $x_i(0)$ and $x_{\alpha_Y(i)}(0)$.
This line segment is contained in the open ball $\{x\in\mathbb{R}^{d}:\|x-y_i\|<\varepsilon(Y)\}$ because $y_i=y_{\alpha_Y(i)}$ holds for $i\in \Nmn$ from \eqref{def:target_set} and \eqref{def:ca_y}, and because both $x_i(0)$ and $x_{\alpha_Y(i)}(0)$ are contained in this open ball.
Hence, both $x_i$ and $x_{\alpha_Y(i)}$ under \eqref{dynamics:displacement} remain in this ball because Lemma \ref{lem:attraction2} always holds on this line segment.
Moreover, $G(X)$ is invariant on $\CA_Y$ because, from \eqref{xyinq}, \eqref{y_delta_in}, and \eqref{y_delta_out}, $G(X)=G(Y)$ holds for any $X\in\CA_Y$.
Then, $\zeta_i(X)$ in \eqref{lambda} is constant.
Hence, from \eqref{dynamics:displacement}, the function $V_i(X):=\|x_i-x_{\alpha_Y(i)}\|^2$ satisfies $\dot{V}_i(X(t))=-(\zeta_i(X(t))+\zeta_{\alpha_Y(i)}(X(t))) V_i(X)$.
By using Lyapunov's theorem \cite{khalil2014nonlinear}, we obtain $\lim_{t\to \infty} \|x_i(t)-x_{\alpha_Y(i)}(t)\|=0$ and its exponential convergence for any $ i\in\Nmn$.
Then, $X(t)\to\CT$ is achieved in the open set $\CA=\cup_{Y\in\CS}\CA_Y\,(\supset \CS)$.
% with $\CS$ satisfying $\bar{\CS}=\CT$. 
Thus, $\CS=\CT\setminus(\CU_1\cup\CU_2)$ is locally attractive, and the convergence is exponential.

\end{document}